\documentclass{article}

\usepackage{arxiv}

\usepackage[utf8]{inputenc} % allow utf-8 input
\usepackage[T1]{fontenc}    % use 8-bit T1 fonts
\usepackage{hyperref}       % hyperlinks
\usepackage{url}            % simple URL typesetting
\usepackage{booktabs}       % professional-quality tables
\usepackage{amsfonts}       % blackboard math symbols
\usepackage{nicefrac}       % compact symbols for 1/2, etc.
\usepackage{microtype}      % microtypography
\usepackage{lipsum}		% Can be removed after putting your text content
\usepackage{graphicx}
\usepackage{natbib}
\usepackage{doi}

\usepackage {amsmath}
\usepackage {amssymb}
\usepackage{amsthm}
\usepackage{xcolor}
\usepackage{tikz}
\usepackage{hyperref}       % hyperlinks
\usepackage{subfig}
\usepackage{mathtools, nccmath}
\usepackage{upgreek}
\usepackage[bb=boondox]{mathalfa}

\usepackage{graphicx}

\newcommand{\olp}{\mathsf{osLip}_{2,P^{1/2}}}
\newcommand{\lp}{\mathsf{Lip}_{\cU \to {2,P^{1/2}}}}

\newcommand{\I}{\mathbb{I}}

\newcommand{\real}{\mathbb{R}}

\renewcommand{\real}{\mathbb{R}}

\newcommand{\E}{\operatorname{\mathbb{E}}}

\newcommand{\bs}{\mathbf{B}}
\newcommand{\e}{\mathrm{e}}
\newcommand{\cU}{\mathcal{U}}

\newcommand{\0}{\mathbb{0}}

\DeclareSymbolFont{bbold}{U}{bbold}{m}{n}
\DeclareSymbolFontAlphabet{\mathbbold}{bbold}

\definecolor{gnblue1}{RGB}{0,36,71}   % 002447  darker
\definecolor{gnblue2}{RGB}{0,60,118}  % 003c76
\definecolor{gnblue3}{RGB}{0,85,164}  % 0055A4
\definecolor{gnblue4}{RGB}{0,108,212} % 006CD4
\definecolor{gnblue5}{RGB}{0,133,255}  % 0085ff
\definecolor{gnblue6}{RGB}{35,156,255} % 239cff
\definecolor{gnblue7}{RGB}{88,177,255} % 58b1ff

\newtheorem{theorem}{Theorem}
\newtheorem{lemma}[theorem]{Lemma}
\newtheorem{corollary}[theorem]{Corollary}
\newtheorem{proposition}[theorem]{Proposition}
\newtheorem{definition}[theorem]{Definition}
\newtheorem{assumption}[theorem]{Assumption}

\newtheorem{remark}{Remark}
\def\red{\hfill $\lhd$}

\definecolor{gnblue4}{RGB}{0,108,212} % 006CD4
\hypersetup{colorlinks=true, linkcolor=gnblue4, breaklinks=true, urlcolor=gnblue4, citecolor=gnblue4}

\title{Incremental Input-to-State Stability and Equilibrium
Tracking for Stochastic Contracting Dynamics}

\author{ \href{https://orcid.org/0000-0001-5066-4700}{\includegraphics[scale=0.06]{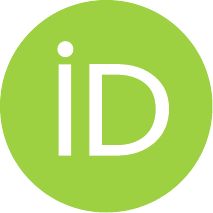}\hspace{1mm}Yu Kawano} \\
	Graduate School of Advanced Science\\ and Engineering,\\
	Hiroshima University,\\
	Higashi-Hiroshima, 739-8527, JP, \\
	\texttt{ykawano[at]hiroshima-u.ac.jp}\thanks{YK is supported in part by JST FOREST Program Grant Number JPMJFR222E.}\\
	\And
	\href{https://orcid.org/0009-0000-3444-0838}{\includegraphics[scale=0.06]{orcid.pdf}\hspace{1mm}Simone Betteti} \\
	Robust and Intelligent Autonomous Systems Lab,\\
	The Italian Institute of Artificial Intelligence\\ for Industry (AI4I),\\
	Torino, 10129, IT, \\
	\texttt{simone.betteti[at]ai4i.it} \\ \\
	\AND
    \href{https://orcid.org/0000-0001-5629-2565}{\includegraphics[scale=0.06]{orcid.pdf}\hspace{1mm}Alexander Davydov} \\
	Department of Mechanical Engineering\\ and Ken Kennedy Institute,\\
	Rice University,\\
	Houston, TX, 77005, US, \\
	\texttt{davydov[at]rice.edu} \\
    \And
    \href{https://orcid.org/0000-0002-4785-2118}{\includegraphics[scale=0.06]{orcid.pdf}\hspace{1mm}Francesco Bullo} \\
	Center for Control, Dynamical Systems and Computation,\\
	University of California at Santa Barbara,\\
	Santa Barbara, CA, 93106 US, \\
	\texttt{bullo[at]ucsb.edu}\thanks{FB is supported in part by
    AFOSR grant FA9550-22-1-0059. FB would like to thank Anand Gokhale for
    insightful conversations.}
}

\date{}

\renewcommand{\undertitle}{Technical Report}
\renewcommand{\shorttitle}{Incremental ISS and Equilibrium
Tracking for Stochastic Contracting Dynamics}

\hypersetup{
pdftitle={Incremental Input-to-State Stability and Equilibrium
Tracking for Stochastic Contracting Dynamics},
pdfauthor={Yu Kawano, Simone Betteti, Alexander Davydov, Francesco Bullo},
pdfkeywords={Stochastic Processes, Noise-input-to-state stability, Equilibrium Tracking},
}

\begin{document}
\maketitle

\begin{abstract}
	In this paper, we study the contractivity of nonlinear stochastic
  differential equations (SDEs) driven by deterministic inputs and Brownian
  motions.  Given a weighted~$\ell_2$-norm for the state space, we show
  that an SDE is incrementally noise- and input-to-state stable if its
  vector field is uniformly contracting in the state and uniformly
  Lipschitz in the input. This result is applied to error estimation for
  time-varying equilibrium tracking in the presence of noise affecting both
  the system dynamics and the input signals. We consider both
  Ornstein–Uhlenbeck processes modeling unbounded noise and Jacobi
  diffusion processes modeling bounded noise. Finally, we turn our
  attention to the associated Fokker–Planck equation of an SDE. For this
  context, we prove incremental input-to-state stability with respect to an
  arbitrary~$p$-Wasserstein metric when the drift vector field is uniformly
  contracting in the state and uniformly Lipschitz in the input with
  respect to an arbitrary norm.
\end{abstract}

% keywords can be removed

\section{Introduction}
Equilibrium tracking is the ability of an input-driven
dynamical system to track an instantaneous input-dependent equilibrium
point as the input varies with time. This property is critically important
for dynamical systems solving optimization problems, where the equilibrium
point corresponds to the optimizer of a time-varying objective
function. Recent work~\citep{AD-VC-AG-GR-FB:23f} has established
equilibrium tracking guarantees for deterministic contracting dynamical
systems. The key result is that deterministic contracting systems can
maintain bounded tracking errors that depend explicitly on the rate of
change of the input signal. These results provide a rigorous foundation for
understanding how optimization algorithms perform in dynamic environments.

In this work, we extend equilibrium tracking theory to stochastic
differential equations, characterizing the effect of noise both in the
system dynamics and in the input signal. While deterministic equilibrium
tracking has been successfully applied to robotic motion planning and
energy management systems, these and other applications inevitably face
stochastic disturbances arising from environmental fluctuations,
measurement noise, and model uncertainties. Understanding how noise
degrades tracking performance is essential for designing robust
optimization algorithms and verifying safe behavior in practical
settings. Our analysis for SDEs quantifies this degradation through
explicit bounds on the expected tracking error as function of both the
intensity of stochastic perturbations and the geometric properties of the
equilibrium manifold.

\textit{Literature review:} Contraction theory has recently established
itself as a computationally-friendly modular robust stability notion for
dynamical systems and algorithms; e.g., see~\citep{WL-JJES:98,FB:26-CTDS}.
An important class of contracting dynamics is given by continuous-time
solvers of convex optimization problems: gradient descent~\citep{RIK:60},
primal-dual dynamics~\citep{GQ-NL:19}, distributed optimization and Nash
seeking dynamics~\citep{AG-AD-FB:23o}, and safe monotone
flows~\citep{AA-JC:25}.  Asymptotic and exponential stability of dynamical
systems solving convex optimization problems (not necessarily contracting)
has long been studied, e.g.,
see~\citep{JW-NE:11,AC-BG-JC:17,JC-SKN:19,NKD-SZK-MRJ:19} among many
others.

A main reference on stochastic contracting dynamical systems is
\citep{QCP-NT-JJES:09}; Theorem~2 and Corollary~1 therein establish
noise-to-state-stability (NSS) for systems without input; see also the
review in~\citep[Section~2.2.2]{HT-SJC-JJES:21}.
Extensions include: \citep{NT-JJS-QCP:10} on how synchronization protects
from noise, \citep[Theorem~1]{NMB-JJES:20} on continuous-time distributed
stochastic gradient, \citep{PHAN:21} on exponential contractivity in mean
square, and \citep{PHAN-LTH:26} on general decay rates.  In a distinct line
of investigation, \cite{SSA-SR:10} define the stochastic logarithmic norm
and \cite{ZA:22} the stochastic logarithmic Lipschitz constant to study
It\^o SDEs with multiplicative noise; see also~\citep{ZA-VS:22}.

Another related set of results concerns the contractivity of the
Fokker-Planck equation with respect to Wasserstein metrics.  When the drift
is the gradient descent of a strongly convex function, the Gibbs
distribution is a fixed point of the Fokker-Planck equation and that the
measure dynamics is contracting with respect to the Wasserstein metric,
e.g., see \cite[Chapter~2]{CV:09} and~\cite[Introduction]{FB-IG-AG:12}.
Recent related works on contractivity in the Wasserstein metric
include~\citep{JB-JJES:19,LC-FH-EM-LJR:25}; specifically,
\cite{JB-JJES:19} investigates the effect of distinct Brownian motions in
Wasserstein space.

Time-varying convex optimization problems have been extensively studied;
see the authoritative reviews \citep{AS-EDA-SP-GL-GBG:20,
  EDA-AS-SB-LM:20}. The work \citep{AD-VC-AG-GR-FB:23f} offers a
comprehensive treatment of equilibrium tracking and continuous-time solvers
for time-varying convex optimization, demonstrating applications to the
design and analysis of safety filters and control barrier functions.
Theoretical advances in equilibrium tracking include: \citep{ZM-FB-AGA:23r}
on the design of control barrier proximal dynamics (CBPD), a contracting
dynamics with a particularly favorable computational structure, and
\citep{ZM-FB-AGA:25ab} on the effects of discretization and quantization of
CBPD in safety filters.
Practical applications include: robot collision avoidance
\citep{MF-SP-VMP-AR:18,AD-VC-AG-GR-FB:23f}, and battery management problems with
electro-thermal constraints \citep{ZM-FB-AGA:23r, ZM-FB-AGA:25p}.

\textit{Contributions:} In Section~\ref{sec:incremental-NISS}, we establish
an incremental noise- and input-to-state stability (NISS) property for
SDEs. Our treatment generalizes known earlier results on (i) the
incremental ISS property of contracting dynamics with
input~\cite[Chapter~3]{FB:26-CTDS}, (ii) NSS of contracting
dynamics~\citep{QCP-NT-JJES:09}, and (iii) the mean-square response of
Ornstein-Uhlenbeck processes.  Given a weighted $\ell_2$ norm on the state
space, we consider vector fields that are uniformly contracting with
respect to the state and uniformly Lipschitz with respect to the input
signal.  For the mean-square distance between trajectories, we provide a
bound valid for all times and we simplify it in the limit of large
times. Additionally, we consider both the case of two stochastic
trajectories and the case of one stochastic and one deterministic
trajectory.

Section~\ref{sec:stoch-eq-tracking} contains the main result of this
paper. This section extends equilibrium tracking theory from deterministic
to stochastic settings, establishing rigorous bounds for tracking errors
when both system dynamics and input signals are subject to noise. We
consider noisy contracting dynamics driven by two distinct noise processes:
Ornstein-Uhlenbeck (OU) processes for unbounded noise and Jacobi diffusion
(JD) processes for bounded noise scenarios. For each noise model, the
analysis covers three increasingly more complex scenarios: (1)
deterministic inputs tracking deterministic equilibrium curves, (2)
stochastic inputs tracking deterministic curves, and (3) stochastic inputs
tracking stochastic curves. This hierarchy provides a complete
characterization of equilibrium tracking under uncertainty.  Our theorems
establish explicit, computable upper bounds on the expected squared
tracking error in terms of fundamental system parameters, including the
contraction rate, the input Lipschitz constant, the multiple relevant noise
intensities, and the rate of change of the reference trajectory. These
bounds reveal how tracking performance degrades with increasing noise and
input variation.

Finally, Section~\ref{sec:incremental-ISS-Wasserstein} considers the
Fokker-Planck equation in Wasserstein metric space. We extend the classic
results on contractivity of the Fokker-Planck equation when the drift is
the gradient descent of a strongly convex function.  Given an arbitrary
norm in the state space, we prove incremental input-to-state stability with
respect to arbitrary $p$-Wasserstein metrics for drift vector fields that
are uniformly contracting in the state and uniformly Lipschitz in the input
signal.

\textit{Paper organization}
We review preliminary concepts about contractivity, equilibrium tracking
and stochastic differential equations in Section~\ref{sec:preliminaries}.
Section~\ref{sec:incremental-NISS} presents our results on incremental
noise- and input-to-state stability. Section~\ref{sec:stoch-eq-tracking}
discusses stochastic equilibrium
tracking. Section~\ref{sec:incremental-ISS-Wasserstein} discusses
incremental input-to-state stability in Wasserstein space.  We provide some
concluding remarks in Section~\ref{sec:conclusions} and proofs of main
theorems in the Appendices.

%%%%%%%%%%%%%%%%%%%%%%%%%%%%%%%%%%%%%%%%%%%%%%%%%%%%%%%%%%%%%%%%%%%%%%%%%%%%%%%%%%%%%%%%%%%%%%%%%%%%%%%%%%%%%%%%%%%%%%%%%%%%%%%%%%%%
%%%%%%%%%%%%%%%%%%%%%%%%%%%%%%%%%%%%%%%%%%%%%%%%%%%%%%%%%%%%%%%%%%%%%%%%%%%%%%%%%%%%%%%%%%%%%%%%%%%%%%%%%%%%%%%%%%%%%%%%%%%%%%%%%%%%

\section{Preliminaries}
\label{sec:preliminaries}
\subsection{Notation}
The field of real numbers is denoted by~$\real$.  The~$n$-component vector
and~$n \times m$-matrix whose all components are~$0$ are denoted by~$\0_n$
and~$\0_{n \times m}$, respectively. The~$n \times n$ identity matrix is
denoted by~$\I_n$.  The Hadamard product (i.e. element-wise product) of two
vectors~$x,y \in \real^n$ is denoted by $x \odot y$.  For vectors~$x,y \in
\real^n$, $x < y$ means the element-wise inequality, i.e., $x_i < y_i$ for
all~$i= 1,\dots, n$. For symmetric matrices~$P, Q \in \real^{n \times n}$, $P \succ
Q$ means that~$P-Q$ is positive definite. For~$P \succ \0_{n
  \times n}$, $\|\cdot\|_{2,P^{1/2}}$ denote the $P$-weighted
$\ell_2$-norm, i.e., $\|x\|_{2,P^{1/2}} := \sqrt{x^\top P x}$ for $x \in
\real^n$. When~$P=\I_n$, the $\ell_2$-norm is denoted by $\| \cdot
\|_2$. Also, the induced matrix $\ell_2$-norm is denoted by $\| \cdot
\|_2$.  The one-sided Lipschitz constant of a mapping~$F:\real^n \to
\real^n$ with respect to norm~$\|\cdot\|_{2,P^{1/2}}$, denoted by~$\olp(F)$
is the minimum constant~$b \in \real$ satisfying~$(F(y)-F(x))^\top P (y-x)
\le b \| y -x \|_{2,P^{1/2}}^2$ for all~$(x,y) \in \real^n \times
\real^n$. The Lipschitz constant of a map~$F:\cU \to \real^n$ from normed
space~$(\| \cdot \|_\cU,\cU)$ to normed space~$(\| \cdot
\|_{2,P^{1/2}},\real^n)$, denoted by~$\lp (F)$ is the minimum
constant~$\ell \in \real$ satisfying~$\|F(v) - F(u)\|_{2,P^{1/2}} \le \ell
\| v- u\|_\cU$ for all~$(u,v) \in \cU \times \cU$.

%%%%%%%%%%%%%%%%%%%%%%%%%%%%%%%%%%%%%%%%%%%%%%%%%%%%%%%%%%%%%%%%%%%%%%%%%%%%%%%%%%%%%%%%%%%%%%%%%%%%%%%%%%%%%%%%%%%%%%%%%%%%%%%%%%%%
%%%%%%%%%%%%%%%%%%%%%%%%%%%%%%%%%%%%%%%%%%%%%%%%%%%%%%%%%%%%%%%%%%%%%%%%%%%%%%%%%%%%%%%%%%%%%%%%%%%%%%%%%%%%%%%%%%%%%%%%%%%%%%%%%%%%

\subsection{Input-to-State Stability of Deterministic Systems and Equilibrium Tracking}\label{sec:ISS_DS}
In this subsection, we recall~\cite[Corollary 3.17]{FB:26-CTDS} for
incremental input-to-state stability (ISS) and equilibrium
tracking~\cite[Theorem~2]{AD-VC-AG-GR-FB:23f} in the deterministic case. For the
sake of self-containedness, we recall these existing result.

Consider an ordinary differential equation (ODE):
\begin{align}\label{eq:ODE}
    \dot x(t) = F(x(t), u(t)).
\end{align}

\begin{proposition}[Incremental input-to-state stability]\label{prop:ISS}
    Given an input-dependent vector field $F:\real^{n} \times \cU \to \real^{n}$ and deterministic measurable inputs $u^{x},u^{y}:\real\to\cU$, where $\cU \subset \real^m$ is compact, consider a pair of systems:
%    \begin{subequations}
    \begin{align*}
        \dot x(t)&=F(x(t), u^{x}(t)),\\
        %--------------------------------------------
        \dot y(t)&=F(y(t), u^{y}(t)).
    \end{align*}        
%    \end{subequations}    
    Assume that there exists a matrix $P=P^{\top}\succ \0_{n \times n}$ such that
    \begin{enumerate}
    \renewcommand{\theenumi}{A\arabic{enumi}}
    
        \item \label{A1:ISS} Contraction in $x$: with respect to the state $x$, the map $F$ is strongly infinitesimally contracting with rate $c>0$ with respect to norm $\|\cdot\|$, uniformly in $u \in \cU$;
        
        \item \label{A2:ISS} Lipschitz in $u$: with respect to the input $u$, the map $F$ from normed space~$(\real^n,\|\cdot\|)$ to normed space~$(\cU, \|\cdot\|_\cU)$ is Lipschitz continuous with constant $\ell>0$, uniformly in $x \in \real^n$;
        
    \end{enumerate}
    Then, for any finite $t \ge 0$, the ODE~\eqref{eq:ODE} is incrementally ISS in the sense that for any trajectories~$x(t)$ and~$y(t)$ from initial conditions~$x(0) \in \real^n$ and~$y(0)\in \real^n$ subject to deterministic measurable inputs $u^{x}$ and $u^{y}$, we have
        \begin{equation}
            \|x(t) - y(t)\|
            %--------------------------------------------
            \leq \|x(0)-y(0)\| \e^{-c t}
            + \ell \int_{0}^{t}\e^{-c (t-\tau)}\|u^{x}(\tau)-u^{y}(\tau)\|_{\cU} \ d\tau.
        \end{equation}
\end{proposition}

We now denote~$u$ by~$\theta$ to emphasize its role as a parameter and consider $\theta$-dependent equilibrium~$x^\star(\theta)$ that is a solution to~$F(x^\star(\theta), \theta) = \0_n$. Under uniform contractivity of~$F$ in $x$, for each~$\theta \in \cU$, $\theta$-dependent equilibrium~$x^\star(\theta)$ exists uniquely. Moreover, under Lipschitzness in $\theta$, we can estimate a tracking error of $x(t) - x^\star(\theta(t))$ as follows.

\begin{proposition}[Equilibrium Tracking]\label{prop:eqt}
    Given~$F:\real^{n} \times \cU \to \real^{n}$, where $\cU \subset \real^m$ is compact, we impose the same assumptions as Proposition~\ref{prop:ISS}. Let~$x^\star(\theta(t))$ be a time-varying equilibrium curve, i.e., $F(x^\star(\theta(t)), \theta(t)) \equiv \0_n$. Then, for any finite $t \ge 0$, for any initial condition~$x(0) \in \real^n$, and continuously differentiable $\theta:\real\to\cU$, the solution to the ODE~\eqref{eq:ODE} satisfies
        % \begin{multline*}
        %     \|x(t) - y(t)\|
        %     %--------------------------------------------
        %     \leq \|x(0)-y(0)\| \e^{-c t}
        %     \\
        %     + \ell \int_{0}^{t}\e^{-c (t-\tau)}\|u^{x}(\tau)-u^{y}(\tau)\|_{\cU} \ d\tau,
        % \end{multline*}
        \begin{equation}
            \|x(t) - x^\star(\theta(t))\|
            %--------------------------------------------
            \leq \|x(0)-x^\star(\theta(0))\| \e^{-c t}
            + \frac{\ell}{c} \int_{0}^{t}\e^{-c (t-\tau)}\|\dot{\theta}(\tau)\|_{\cU} \ d\tau,
        \end{equation}
    and
        \begin{equation}\label{eq:eqt_dtr}
        \limsup_{t \to \infty} \|x(t) - x^\star(\theta(t))\|
        %--------------------------------------------
        \le \frac{\ell}{c^2} \limsup_{t \to \infty} \|\dot{\theta}(\tau)\|_{\cU}.
\end{equation}
\end{proposition}

%%%%%%%%%%%%%%%%%%%%%%%%%%%%%%%%%%%%%%%%%%%%%%%%%%%%%%%%%%%%%%%%%%%%%%%%%%%%%%%%%%%%%%%%%%%%%%%%%%%%%%%%%%%%%%%%%%%%%%%%%%%%%%%%%%%%
%%%%%%%%%%%%%%%%%%%%%%%%%%%%%%%%%%%%%%%%%%%%%%%%%%%%%%%%%%%%%%%%%%%%%%%%%%%%%%%%%%%%%%%%%%%%%%%%%%%%%%%%%%%%%%%%%%%%%%%%%%%%%%%%%%%%

\subsection{Stochastic Differential Equations}

In this paper, our goal is to generalize Propositions~\ref{prop:ISS} and~\ref{prop:eqt} for the ODE~\eqref{eq:ODE} to the stochastic differential equation (SDE) in the so-called It\^o sense:
\begin{align}\label{eq:SDE}
    dx_{t} =F(x_{t}, u(t)) dt + \Sigma (x_{t}, u(t)) d\bs_{t},
\end{align}
where $\Sigma: \real^{n} \times \cU \to \real^{n \times r}$ is a state- and input-dependent matrix, and $\bs_t$ is an $r$-dimensional Brownian motion. It is customary to introduce the following standard assumption for the Lipschitzness and linear growth.

\begin{assumption}[Lipschitzness and linear growth]\label{asm:lip}
    In the SDE~\eqref{eq:SDE}, the drift vector field~$F$ and the dispersion matrix~$\Sigma$ satisfy the global Lipschitz and linear growth conditions:
    \begin{subequations}
    \begin{align}
        &\|F(x, u) - F(y, u) \|_2  + \| \Sigma (x, u) - \Sigma (y, u) \|_{\mathsf{F}} \le L \| x - y\|_2,
        \label{eq:lip}\\
        %--------------------------------------------
        &\|F(x, u)\|_2^2 + \| \Sigma (x, u) \|_{\mathsf{F}}^2 \le L (1 + \|x\|_2^2)
        \label{eq:lg}
    \end{align}        
    \end{subequations}
    for each $x, y \in \real^n$ and every $u \in \cU$, where $L$ is a positive constant, and $\| X \|_{\mathsf{F}}^2 = \operatorname{trace}(X^\top X)$ is the Frobenius norm of the matrix $X$.
    \red
\end{assumption}

Assumption~\ref{asm:lip} of Lipschitz continuity and linear growth for the drift vector field $F$ and the dispersion matrix $\Sigma$ guarantees the unique existence of a strong solution to~\eqref{eq:SDE}, where we refer to~\cite[Section 5.2]{IK-SES:14} for the definition of a strong solution.

\begin{proposition}[Existence, uniqueness, and joint measurability of strong solutions]\label{prop:sol}
    For an SDE~\eqref{eq:SDE} with globally Lipschitz and linearly growing drift and dispersion (Assumption~\ref{asm:lip}), there exists the unique strong solution $\{x_t\}_{t \ge 0}$ with continuous sample paths and with finite second moment on finite horizons such that
    \begin{subequations}\label{eq:bd}
    \begin{align}
        &\E^{x_0} [\| x_t \|_2] \le (1 + \| x_0 \|_2) \e^{\frac{(1+L)t}{2}}, \quad \forall t \ge 0, 
        \label{eq1:bd}\\
        %--------------------------------------------
        &\E^{x_0} [\| x_t \|_2^2] \le (1 + \| x_0 \|_2^2) \e^{(1+L)t}, \quad \forall t \ge 0,
        \label{eq2:bd}
    \end{align}        
    \end{subequations}
    where $\E^{x_0} [ \cdot ]$ denotes the conditional expectation given $x_0$. Moreover, $\sup_{\tau \in [0, t]} \E [\| x_\tau \|^2] < + \infty$ for each finite $t > 0$ if the initial condition $x_0$ is square integrable, that is, $\E [ \| x_0\|^2 ] < \infty$.~\red
\end{proposition}

We next introduce two central tools in the analysis of SDEs: the infinitesimal generator~\cite[Theorem~7.3.3]{BO:13} and Dynkin's formula~\cite[Theorem~7.4.1]{BO:13}.

\begin{definition}[Infinitesimal generator]
    For $\varphi: \real^{n}\rightarrow\real$ of class $C^2$, the infinitesimal generator $\mathcal{L}$ of the SDE~\eqref{eq:SDE} acting on $\varphi$ is given by
    \begin{align}\label{eq:inf_gen}
        \mathcal{L} \varphi (x, u)
        %--------------------------------------------
        &= \frac{\partial \varphi  (x)}{\partial x} F(x, u) 
         +\frac{1}{2} \operatorname{trace}\bigl(\Sigma (x, u)^{\top} \operatorname{Hess} ( \varphi(x) )  \Sigma (x, u)\bigr).
    \end{align}
\end{definition}

\begin{proposition}[Dynkin’s formula]\label{prop:Dynkin}
Consider an SDE~\eqref{eq:SDE} with globally Lipschitz and linearly growing drift and dispersion (Assumption~\ref{asm:lip}) and with square-integrable initial condition $x_0$, that is, $\E [\| x_0 \|^2] < \infty$. Let~$\varphi: \real^{n}\rightarrow\real$ be of class $C^2$, $\mathcal{L}$ be the infinitesimal generator of the SDE~\eqref{eq:SDE} acting on~$\varphi$, and $t > 0$ be a finite time. Then, the unique strong solution~$\{x_t\}_{t \ge 0}$ to the SDE satisfies
\begin{align*}
    \E^{x_0}[\varphi(x_t)] 
    %--------------------------------------------
    = \varphi (x_0) + \E^{x_0} \left[\int_0^t (\mathcal{L} \varphi) (x_\tau, u(\tau)) d \tau\right],
\end{align*}
where $\E^{x_0}[\cdot]$ denotes the conditional expectation given $x_0$.~\red
\end{proposition}

%%%%%%%%%%%%%%%%%%%%%%%%%%%%%%%%%%%%%%%%%%%%%%%%%%%%%%%%%%%%%%%%%%%%%%%%%%%%%%%%%%%%%%%%%%%%%%%%%%%%%%%%%%%%%%%%%%%%%%%%%%%%%%%%%%%%
%%%%%%%%%%%%%%%%%%%%%%%%%%%%%%%%%%%%%%%%%%%%%%%%%%%%%%%%%%%%%%%%%%%%%%%%%%%%%%%%%%%%%%%%%%%%%%%%%%%%%%%%%%%%%%%%%%%%%%%%%%%%%%%%%%%%

\subsection{Ornstein-Uhlenbeck Process}
An example of SDEs satisfying Assumption~\ref{asm:lip} of Lipschitz continuity and linear growth for the drift vector field $F$ and the dispersion matrix $\Sigma$ is an \emph{Ornstein-Uhlenbeck} (OU) process:
\begin{align}\label{eq:SDE_OU}
    d x_t = - c x_t dt + \frac{\sigma}{\sqrt{n}}  d\bs_{t},
\end{align}
where~$c, \sigma > 0$. The mean and expected squared norm of~$x_t$ are
\begin{align}
    \E [x_t] &= \E [x_0] \e^{-c t}, \nonumber\\
    \E [\|x_t\|_2^2] &= \e^{-2ct} \| x_0\|^2 + \frac{\sigma^2}{2c} (1 - \e^{-2ct}).
    \label{eq:SDE_OU_var}
\end{align}

%%%%%%%%%%%%%%%%%%%%%%%%%%%%%%%%%%%%%%%%%%%%%%%%%%%%%%%%%%%%%%%%%%%%%%%%%%%%%%%%%%%%%%%%%%%%%%%%%%%%%%%%%%%%%%%%%%%%%%%%%%%%%%%%%%%%
%%%%%%%%%%%%%%%%%%%%%%%%%%%%%%%%%%%%%%%%%%%%%%%%%%%%%%%%%%%%%%%%%%%%%%%%%%%%%%%%%%%%%%%%%%%%%%%%%%%%%%%%%%%%%%%%%%%%%%%%%%%%%%%%%%%%

\subsection{Jacobi Diffusion Process}
Even if Assumption~\ref{asm:lip} of Lipschitz continuity and linear growth for the drift vector field~$F$ and the dispersion matrix~$\Sigma$ does not hold, we can still work with weak solutions. Consider the following SDE:
\begin{align}\label{eq:SDE_JDP}
    d x_t = - c (x_t - \theta) dt + \sigma \operatorname{diag}(x_t \odot (a -x_t))^\frac{1}{2} d\bs_{t},
\end{align}
where $c > 0$ and $\theta, a > \0_n$. Each component of the SDE~\eqref{eq:SDE_JDP} is called a Jacobi diffusion (JD), and its stationary distribution is a $\beta$-distribution~\citep{SC-RK:21}.

The JD~\eqref{eq:SDE_JDP} does not satisfy the Lipschitz assumption~\eqref{eq:lip}, but does the linear growth assumption~\eqref{eq:lg}. Thus, the JD~\eqref{eq:SDE_JDP} admits a weak solution~\cite[Problem 3.15]{IK-SES:14}, which satisfies~\eqref{eq:bd}. Moreover, $\sup_{\tau \in [0, t]} \E [\| x_\tau \|^2] < + \infty$ for each finite $t > 0$ if the initial condition $x_0$ is square integrable, that is, $\E [ \| x_0\|^2 ] < \infty$. Also, by Feller's test for explosions~\cite[Theorem 5.29]{IK-SES:14}, the weak solution stay in the the interval $(\0_m, a) := (0, a_1) \times \cdots \times (0, a_m)$ almost surely if 
\begin{align}\label{eq:Feller}
    \frac{\sigma_\xi^2}{2c} a \le \theta \le \left( 1 - \frac{\sigma_\xi^2}{2c}\right) a,
    \quad \sigma_\xi^2 < c
\end{align}
Therefore, Dynkin’s formula can be applied on $(\0_m, a)$. From~\eqref{eq:Feller}, a JD process is helpful when modeling a bounded noise.

%%%%%%%%%%%%%%%%%%%%%%%%%%%%%%%%%%%%%%%%%%%%%%%%%%%%%%%%%%%%%%%%%%%%%%%%%%%%%%%%%%%%%%%%%%%%%%%%%%%%%%%%%%%%%%%%%%%%%%%%%%%%%%%%%%%%
%%%%%%%%%%%%%%%%%%%%%%%%%%%%%%%%%%%%%%%%%%%%%%%%%%%%%%%%%%%%%%%%%%%%%%%%%%%%%%%%%%%%%%%%%%%%%%%%%%%%%%%%%%%%%%%%%%%%%%%%%%%%%%%%%%%%

%%%%%%%%%%%%%%%%%%%%%%%%%%%%%%%%%%%%%%%%%%%%%%%%%%%%%%%%%%%%%%%%%%%%%%%%%%%%%%%%%%%%%%%%%%%%%%%%%%%%%%%%%%%%%%%%%%%%%%%%%%%%%%%%%%%%
%%%%%%%%%%%%%%%%%%%%%%%%%%%%%%%%%%%%%%%%%%%%%%%%%%%%%%%%%%%%%%%%%%%%%%%%%%%%%%%%%%%%%%%%%%%%%%%%%%%%%%%%%%%%%%%%%%%%%%%%%%%%%%%%%%%%

%%%%%%%%%%%%%%%%%%%%%%%%%%%%%%%%%%%%%%%%%%%%%%%%%%%%%%%%%%%%%%%%%%%%%%%%%%%%%%%%%%%%%%%%%%%%%%%%%%%%%%%%%%%%%%%%%%%%%%%%%%%%%%%%%%%%
%%%%%%%%%%%%%%%%%%%%%%%%%%%%%%%%%%%%%%%%%%%%%%%%%%%%%%%%%%%%%%%%%%%%%%%%%%%%%%%%%%%%%%%%%%%%%%%%%%%%%%%%%%%%%%%%%%%%%%%%%%%%%%%%%%%%

\section{Incremental Noise- and Input-to-State Stability}
\label{sec:incremental-NISS}
We present the first main result of this paper for contracting
SDEs. Specifically, we establish an incremental noise- and input-to-state
stability (NISS) property for SDEs as a generalization of
Proposition~\ref{prop:ISS} for ISS of ODEs.

\begin{theorem}[Incremental noise- and input-to-state stability]\label{thm:NISS}
    Given an input-dependent vector field $F:\real^{n} \times \cU \to \real^{n}$, a state- and input-dependent matrix $\Sigma : \real^{n} \times \cU \to \real^{n \times r}$, and deterministic measurable inputs $u^{x},u^{y}:\real\to\cU$, where $\cU \subset \real^m$ is compact, consider the realizations with independent noises of
%    \begin{subequations}
    \begin{align*}
        dx_{t}&=F(x_{t}, u^{x}(t)) dt + \Sigma (x_{t}, u^{x}(t))d\bs_{t}^x,\\
        %--------------------------------------------
        dy_{t}&=F(y_{t}, u^{y}(t)) dt + \Sigma (y_{t}, u^{y}(t))d\bs_{t}^y.
    \end{align*}        
%    \end{subequations}
    We impose Assumption~\ref{asm:lip} of Lipschitz continuity and linear growth for the drift vector field~$F$ and the dispersion matrix~$\Sigma$.
    We additionally assume that there exists a matrix $P=P^{\top}\succ \0_{n \times n}$ such that
    \begin{enumerate}
    \renewcommand{\theenumi}{A\arabic{enumi}}
    
        \item \label{A1:NISS} Contraction in $x$: there exists $c>0$ such that $\olp (F) \le - c$, uniformly in $u \in \cU$;
        
        \item \label{A2:NISS} Lipschitz in $u$: there exists $\ell>0$ such that $\lp (F) \le \ell$, uniformly in $x \in \real^n$;
        
        \item \label{A3:NISS} Bounded dispersion: the dispersion is uniformly bounded in the $P$-weighted Frobenius norm, namely
        \begin{align}\label{eq1:A3:NISS}
            \sigma_x^{2}
            %--------------------------------------------
            :=\sup_{(x,u) \in \real^n \times \cU} \|\Sigma (x,u)\|^{2}_{\mathsf{F},P^{1/2}} 
             =\sup_{(x,u) \in \real^n \times \cU} \operatorname{trace}\bigl(\Sigma (x,u)^{\top} P \Sigma (x,u)\bigr)<+\infty.
        \end{align}
    \end{enumerate}
    Then, for each $\alpha\in(0,1)$ and any finite $t \ge 0$,
    \begin{enumerate}
        \item \label{itm1:NISS} the SDE~\eqref{eq:SDE} is incrementally NISS in the sense that, for any two realizations $x_t$ and $y_t$ from random square integrable initial conditions $x_0$ and $y_0$ and subject to deterministic measurable inputs $u^{x}$ and $u^{y}$ and different realizations of the Brownian motion, we have
        \begin{align}\label{eq1:NISS}
            &\E \bigl[ \|x_{t} - y_{t}\|^{2}_{2,P^{1/2}} \bigr]
            \nonumber\\
            %--------------------------------------------
            &\quad \leq \E \bigl[ \|x_{0}-y_{0}\|^{2}_{2,P^{1/2}} \bigr]\e^{-2c\alpha t}
            + \frac{1}{\alpha} \frac{\sigma_x^{2}}{c}(1-\e^{-2c\alpha t})
            +\frac{1}{1-\alpha}\frac{\ell^{2}}{2c}\int_{0}^{t}\e^{-2c\alpha(t-\tau)}\|u^{x}(\tau)-u^{y}(\tau)\|_{\cU}^{2}\ d\tau,
        \end{align}
        and, without assuming the measurability of $u^{x}$ or $u^{y}$,
        \begin{align}\label{eq2:NISS}
            &\limsup_{t \to \infty} \; \E \bigl[ \|x_{t} - y_{t}\|^{2}_{2,P^{1/2}} \bigr]
             \leq \frac{1}{\alpha} \frac{\sigma_x^{2}}{c} + \frac{1}{\alpha (1-\alpha)} \frac{\ell^{2}}{4c^2} \limsup_{t \to \infty} \|u^{x}(t)-u^{y}(t)\|_{\cU}^{2};
        \end{align}

        \item \label{itm2:NISS}  for any solution $y(t)$ to the ODE~\eqref{eq:ODE} and any realization $x_t$ of the SDE~\eqref{eq:SDE} from random square integrable initial condition $x_0$ under deterministic measurable inputs $u^y$ and $u^x$, we have
        \begin{align}\label{eq3:NISS}
            &\E \bigl[ \|x_{t} - y(t)\|^{2}_{2,P^{1/2}} \bigr]
            \nonumber\\
            %--------------------------------------------
            &\quad \leq \E \bigl[ \|x_{0}-y(0)\|^{2}_{2,P^{1/2}} \bigr]\e^{-2c\alpha t}
            +\frac{1}{\alpha}\frac{\sigma_x^{2}}{2c}(1-\e^{-2c\alpha t})
            +\frac{1}{1-\alpha} \frac{\ell^{2}}{2c}\int_{0}^{t}\e^{-2c\alpha(t-\tau)}\|u^{x}(\tau)-u^y(\tau)\|_{\cU}^{2}\ d\tau,
        \end{align}
        and, without assuming the measurability of $u^{x}$ or $u^y$,
        \begin{align*}
            &\limsup_{t \to \infty} \; \E \bigl[ \|x_{t} - y(t)\|^{2}_{2,P^{1/2}} \bigr]
             \leq \frac{1}{\alpha} \frac{\sigma_x^{2}}{2c} + \frac{1}{\alpha (1-\alpha)} \frac{\ell^{2}}{4c^2} \limsup_{t \to \infty} \|u^{x}(t)-u^y(t)\|_{\cU}^{2}.
        \end{align*}
    \end{enumerate}
\end{theorem}

\begin{proof}
(Proof of item~\ref{itm1:NISS})
we consider the combined SDE:
    \begin{align}\label{pf1:NISS}
        \begin{bmatrix}
            dx_{t}\\
            dy_{t}
        \end{bmatrix}
        %--------------------------------------------
        &=
        \begin{bmatrix}
            F(x_{t}, u^{x}(t))\\
            F(y_{t}, u^{y}(t))
        \end{bmatrix} dt + 
        \begin{bmatrix}
            \Sigma (x_{t}, u^{x}(t)) & \0_{n \times r}\\
            \0_{n \times r} & \Sigma (y_{t}, u^{y}(t))
        \end{bmatrix} d\bs_{t},
    \end{align}
    where $\bs_{t}$ is a $2r$ dimensional Brownian motion, i.e., we consider independent noise. For the combined SDE~\eqref{pf1:NISS}, we consider a class $C^2$ function $\varphi (x,y) := \| x -y \|_{2,P^{1/2}}^2$ and compute the Hessian matrix:
    \begin{align*}
        \operatorname{Hess} \bigl(\| x -y \|_{2,P^{1/2}}^2\bigr) 
        %--------------------------------------------
        &= \operatorname{Hess} \left(\begin{bmatrix}
            x\\
            y
        \end{bmatrix}^{\top}
        \begin{bmatrix}
            P & -P\\
            -P & P
        \end{bmatrix}
        \begin{bmatrix}
            x\\
            y
        \end{bmatrix}\right)
        %--------------------------------------------
        = 2\begin{bmatrix}
            P & -P\\
            -P & P
        \end{bmatrix},
    \end{align*}
    and thus, from the bounded dispersion assumption (Assumption~\ref{A3:NISS})    
    \begin{align}\label{pf2:NISS}
        &\operatorname{trace}\Biggl(\begin{bmatrix}
            \Sigma (x,u^x) & \0_{n \times r}\\
            \0_{n \times r} & \Sigma (y,u^y)
        \end{bmatrix}^\top 
        \begin{bmatrix}
                P & -P\\
                -P & P
            \end{bmatrix}
        \begin{bmatrix}
            \Sigma (x,u^x) & \0_{n \times r}\\
            \0_{n \times r} & \Sigma (y,u^y)
        \end{bmatrix}\Biggr)
        \nonumber\\
        %--------------------------------------------
        &\quad =\operatorname{trace} (\Sigma (x,u^x)^{\top} P \Sigma (x,u^x)) 
          + \operatorname{trace}(\Sigma (y,u^y)^{\top} P \Sigma (y,u^y))
        \nonumber\\
        %--------------------------------------------
        &\quad \leq 2 \sup_{(x,u^x) \in \real^n \times \cU} \operatorname{trace}(\Sigma (x,u^x)^{\top} P \Sigma (x,u^x))
         =2\sigma_x^{2}.
    \end{align}
    Given the contractivity (with respect to $x$) and Lipschitzness (with respect to $u$) assumptions (Assumptions~\ref{A1:NISS} and~\ref{A2:NISS}) on~$F$ and from~\eqref{eq:inf_gen} and~\eqref{pf2:NISS}, the infinitesimal generator on the function $\varphi (x,y)= \| x -y \|_{2,P^{1/2}}^2$ satisfies
    \begin{align}\label{pf3:NISS}
        \mathcal{L} \| x_t -y_t \|_{2,P^{1/2}}^2 
        %--------------------------------------------
        &\le - 2c \| x_t -y_t \|_{2,P^{1/2}}^2  
         + 2 \ell \| x_t -y_t \|_{2,P^{1/2}} \| u^x(t) - u^y(t) \|_{\cU} 
         + 2 \sigma_x^2 \; \mbox{ a.s.}
    \end{align}

    Next, we apply Dynkin's formula in Proposition~\ref{prop:Dynkin}. For any finite $t \ge s \ge 0$, taking the conditional expectation given $x_s$ and $y_s$ leads to
    \begin{align}\label{pf4:NISS}
        &\E^{x_s, y_s} \bigl[ \| x_t -y_t \|_{2,P^{1/2}}^2 \bigr] 
        - \| x_s - y_s \|_{2,P^{1/2}}^2 
        \nonumber\\
       %--------------------------------------------
       &\quad =  \E^{x_s, y_s} \left[\int_s^t \mathcal{L} \| x_\tau - y_\tau \|_{2,P^{1/2}}^2 d \tau\right] 
       \nonumber\\
       %--------------------------------------------
       &\quad \le - 2c \E^{x_s, y_s} \left[\int_s^t  \| x_\tau -y_\tau \|_{2,P^{1/2}}^2 d \tau\right] 
       \nonumber\\
        &\quad \quad + 2 \ell \E^{x_s, y_s} \left[\int_s^t  \| x_\tau -y_\tau \|_{2,P^{1/2}} \| u^x(\tau) - u^y(\tau) \|_{\cU}  d \tau\right]
         + \int_s^t 2 \sigma_x^2 d \tau \; \mbox{ a.s.},
    \end{align}
    where the inequality follows from \eqref{pf3:NISS}. 
    
    From~\eqref{eq1:bd} and the triangular inequality, it follows that
    \begin{align*}
    \int_s^t \E^{x_s, y_s} [\| x_\tau -y_\tau \|_2] d \tau
     \le  \int_s^t (\E^{x_s} [\| x_\tau \|_2] + \E^{y_s} [\| y_\tau \|_2]) d \tau
     \le  (2 + \| x_0 \|_2 + \| y_0 \|_2)
    \int_s^t \e^{\frac{(1+L)\tau}{2}} d \tau \; \mbox{ a.s.}
    \end{align*}
    Thus, $\int_s^t \E^{x_s, y_s} [ \| x_\tau -y_\tau \|_{2,P^{1/2}} ] d \tau$ is finite almost surely. Similarly, from~\eqref{eq2:bd} and $\| x -y \|_2^2 \le 2 (\| x\|_2^2 + \|y \|_2^2)$, $\int_s^t \E^{x_s, y_s} [ \| x_\tau -y_\tau \|_{2,P^{1/2}}^2 ] d \tau$ is finite almost surely. Therefore, from Tonelli's theorem for measurable non-negative functions~\cite[Theorems 4.4-4.5]{HB:11}, we can exchange the order of taking the time-integration and expectation in~\eqref{pf4:NISS}. Namely, we have
    \begin{align}\label{pf5:NISS}
        \E^{x_s, y_s} \bigl[ \| x_t -y_t \|_{2,P^{1/2}}^2 \bigr] 
        - \| x_s - y_s \|_{2,P^{1/2}}^2 
         \le& - 2c \int_s^t \E^{x_s, y_s} \bigl[ \| x_\tau -y_\tau \|_{2,P^{1/2}}^2 \bigr] d \tau
       \nonumber\\
        & + 2 \ell \int_s^t \| u^x(\tau) - u^y(\tau) \|_{\cU} \E^{x_s, y_s} \bigl[\| x_\tau -y_\tau \|_{2,P^{1/2}} \bigr]  d \tau
        \nonumber\\
        & + \int_s^t 2 \sigma_x^2 d \tau \; \mbox{ a.s.}
    \end{align}

    Since $\sqrt{\cdot}$ is a concave function, Jensen's inequality~\cite[Proposition 4.9]{HB:11} yields
    \begin{align}\label{pf6:NISS}
    &\E^{x_s, y_s} [\| x_\tau -y_\tau \|_{2,P^{1/2}}] 
     \le \sqrt{\E^{x_s, y_s} \bigl[\| x_\tau -y_\tau \|_{2,P^{1/2}}^2\bigr] }.
    \end{align}
    Denoting $g(t) := \E^{x_s, y_s} \bigl[\| x_t -y_t \|_{2,P^{1/2}}^2 \bigr]$, \eqref{pf5:NISS} and \eqref{pf6:NISS} lead to
    \begin{align*}
        g(t) - g(s)
        %--------------------------------------------
        &\le \int_s^t \Bigl( -2 c g(\tau) + 2 \sigma_x^2
        + 2 \ell \| u^x(\tau) - u^y(\tau) \|_{\cU} \sqrt{g(\tau)} \Bigr) d\tau \; \mbox{ a.s.},
    \end{align*}
    and consequently,
    \begin{align*}
        \frac{g(t) - g(s)}{t-s} 
        %--------------------------------------------
        &\le  \frac{1}{t-s} \int_s^t \Bigl( -2 c g(\tau)  + 2 \sigma_x^2   
        + 2 \ell \| u^x(\tau) - u^y(\tau) \|_{\cU} \sqrt{g(\tau)} \Bigr) d\tau \; \mbox{ a.s.}
    \end{align*}
    for any $t > s$. Since $g(t)$ is continuous, its Dini derivative~\cite[Chapter 2.4.2]{FB:26-CTDS} is 
    \begin{align*}
        D^+ g(s) 
        %--------------------------------------------
        := \lim_{t \to s^+} \frac{g(t) - g(s)}{t-s} 
        \le  -2 c g(s) + 2 \sigma_x^2 
        + 2 \ell \| u^x(s) - u^y(s) \|_{\cU} \sqrt{g(s)}  \; \mbox{ a.s.}
    \end{align*}
    Using non-negativity of $g(s) = \| x_s - y_s \|_{2,P^{1/2}}^2$, compute by using Young's inequality: 
    \begin{align}\label{eq:alp}
        &2 \ell \| u^x(s) - u^y(s) \|_{\cU} \sqrt{g(s)} \le 2c (1-\alpha) g(s) + \frac{\ell^2}{2c (1-\alpha)} \| u^x(s) - u^y(s) \|_{\cU}^2 \; \mbox{ a.s.},
    \end{align}
    where $\alpha \in (0,1)$, and consequently,
    \begin{align*}
        D^+ g(s) 
        %--------------------------------------------
        & \le -2 c \alpha g(s) + 2 \sigma_x^2 + \frac{\ell^2}{2c (1-\alpha)} \| u^x(s) - u^y(s) \|_{\cU}^2 \; \mbox{ a.s.}
    \end{align*}
    Since this holds for arbitrary finite $s \ge 0$, the comparison principle~\cite[Lemma 3.4]{HKK:02} yields
    \begin{align*}
        \| x_t - y_t \|_{2,P^{1/2}}^2 &\le \| x_0 - y_0 \|_{2,P^{1/2}}^2 \e^{-2 c \alpha t} + \frac{\sigma_x^2}{c \alpha} (1 - \e^{-2 c \alpha t})
        \nonumber\\
        &\quad + \frac{\ell^2}{2c (1-\alpha)} \int_0^t \e^{-2 c \alpha (t-\tau)}  \| u^x(\tau) - u^y(\tau) \|_{\cU}^2 d \tau \; \mbox{ a.s.},
    \end{align*}
    where recall the compactness of $\cU$ and the measurability of both $u^x$ and $u^y$, guaranteeing the boundedness of the integral. Taking the expectation concludes~\eqref{eq1:NISS}. 

    To show \eqref{eq2:NISS}, we compute an upper bound on the right-hand side of~\eqref{eq1:NISS}, yielding
    \begin{align*}
            \E \bigl[ \|x_{t} - y_{t}\|^{2}_{2,P^{1/2}} \bigr]
            \leq& \E \bigl[ \|x_{0}-y_{0}\|^{2}_{2,P^{1/2}} \bigr]\e^{-2c\alpha t}+\frac{\sigma_x^{2}}{c\alpha}(1-\e^{-2c\alpha t})
            \nonumber\\
            & +\frac{\ell^{2}}{4c^2 \alpha(1-\alpha)} (1-\e^{-2c\alpha t}) \sup_{\tau \in [0,t]} \|u^{x}(\tau)-u^{y}(\tau)\|_{\cU}^{2}.
    \end{align*}
    This holds for any finite $t \ge 0$. The compactness of $\cU$ implies that $\limsup_{t \to \infty} \|u^{x}(\tau)-u^{y}(\tau)\|_{\cU}^{2}$ is always finite even without assuming the measurability of $u^x$ or $u^y$. Also, from the square integrability assumption of the random initial conditions $x_0$ and $y_0$, the right-hand side is finite for any $t \ge 0$. Therefore, we can take the limit superiors of both sides, leading to \eqref{eq2:NISS}.

%%%%%%%%%%%%%%%%%%%%%%%%%%%%%%%%%%%%%%%%%%%%%%%%%%%%%%%%%%%%%%%%%%
%%%%%%%%%%%%%%%%%%%%%%%%%%%%%%%%%%%%%%%%%%%%%%%%%%%%%%%%%%%%%%%%%%

(Proof of item~\ref{itm2:NISS})
    Instead of~\eqref{pf1:NISS}, consider 
    \begin{align*}
        \begin{bmatrix}
            dx_{t}\\
            dy_{t}
        \end{bmatrix}
        %--------------------------------------------
        =
        \begin{bmatrix}
            F(x_{t}, u^{x}(t))\\
            F(y_{t}, u^y(t))
        \end{bmatrix} dt 
        + 
        \begin{bmatrix}
            \Sigma (x_{t}, u^x) & \0_{n \times r}\\
            \0_{n \times r} & \0_{n \times r}
        \end{bmatrix} d\bs_{t}.
    \end{align*}
    Also, instead of~\eqref{pf2:NISS}, compute
    \begin{align*}
        &\operatorname{trace}\Biggl(\begin{bmatrix}
            \Sigma (x, u^x) & \0_{n \times r}\\
            \0_{n \times r} & \0_{n \times r}
        \end{bmatrix}^\top
        \begin{bmatrix}
                P & -P\\
                -P & P
            \end{bmatrix}
        \begin{bmatrix}
            \Sigma (x, u^x) & \0_{n \times r}\\
            \0_{n \times r} & \0_{n \times r}
        \end{bmatrix}\Biggr)
        \nonumber\\
        %--------------------------------------------
        &\quad =\operatorname{trace} (\Sigma (x, u^x)^{\top} P \Sigma (x, u^x))
         \leq \sup_{(x, u^x) \in \real^n \times \cU} \operatorname{trace}(\Sigma (x, u^x)^{\top} P \Sigma (x, u^x))
        %--------------------------------------------
        = \sigma_x^{2}.
    \end{align*}
    Thus, instead of~\eqref{pf3:NISS}, we have
    \begin{align*}
        \mathcal{L} \| x_t - z_t \|_{2,P^{1/2}}^2 
        %--------------------------------------------
        &\le - 2c \| x_t - z_t \|_{2,P^{1/2}}^2  
        + 2 \ell \| x_t - z_t \|_{2,P^{1/2}} \| u^x(t) - u^y(t) \|_{\cU} 
        + \sigma_x^2 \; \mbox{ a.s.}
    \end{align*}
    The rest of the proof is similar to that of item~\ref{itm1:NISS}.
\end{proof}

\begin{remark}[Comparisons]
  \begin{enumerate}
  \renewcommand{\theenumi}{\roman{enumi}}
  
  \item In the deterministic case when $\sigma^x = 0$,
    Theorem~\ref{thm:NISS} reduces to Proposition~\ref{prop:ISS} for ISS
    (note that Theorem~\ref{thm:NISS} uses the squared norm to evaluate the
    second moment).  The additional parameter $\alpha \in (0,1)$ arises
    from the application of Young’s inequality~\eqref{eq:alp} for
    estimating an upper bound under the presence of both noise and
    deterministic inputs. The term~$\frac{1}{\alpha}
    \frac{\sigma_x^{2}}{c}(1-\e^{-2c\alpha t})$ in~\eqref{eq1:NISS}
    represents the effect of noise and is similar to the expected squared
    norm of an OU process~\eqref{eq:SDE_OU_var}. In the denominator, we
    have~$c$ instead of~$2c$ because we consider two realizations~$x_t$
    and~$y_t$ of SDE~\eqref{eq:SDE}. In fact, in~\eqref{eq3:NISS}
    considering a realization~$x_t$ of SDE~\eqref{eq:SDE}, the term
    representing the effect of noise is~$\frac{1}{\alpha}
    \frac{\sigma_x^{2}}{2c}(1-\e^{-2c\alpha t})$.
  \item In the stochastic case without deterministic
    input $u$, or more precisely when $u^x(t) \equiv u^y(t)$,
    Theorem~\ref{thm:NISS} recovers \cite[Theorem~2]{QCP-NT-JJES:09} for
    noise-to-state stability in the constant metric case. 
  \end{enumerate}
  In summary, Theorem~\ref{thm:NISS} can be viewed as a generalization of
  Proposition~\ref{prop:ISS} and \cite[Theorem~2]{QCP-NT-JJES:09}.
  \red
\end{remark}

%%%%%%%%%%%%%%%%%%%%%%%%%%%%%%%%%%%%%%%%%%%%%%%%%%%%%%%%%%%%%%%%%%
%%%%%%%%%%%%%%%%%%%%%%%%%%%%%%%%%%%%%%%%%%%%%%%%%%%%%%%%%%%%%%%%%%

\section{Stochastic Equilibrium Tracking}
\label{sec:stoch-eq-tracking}
In this section, we generalize Proposition~\ref{prop:eqt} for equilibrium tracking to the stochastic setting. As discussed in the Introduction, equilibrium tracking has been applied in several deterministic applications. In practice, however, both systems and signals are subject to stochastic fluctuations and noise. To address this, we derive bounds for equilibrium tracking performance in the presence of stochastic disturbances. As stochastic disturbances, we consider two different noises driven by A) OU process; and B) JD process, where JD process is useful to represent bounded noise.

\subsection{Driven by Ornstein-Uhlenbeck Processes}
We consider equilibrium tracking for noisy contracting dynamics:
\begin{align}\label{eq1:SDE_OU_SI}
    dx_{t}=F(x_{t}, u_t) dt + \Sigma (x_{t}, u_t) d\bs_{t}^x,
\end{align}
where contraction rate is~$\olp (F) \le - c < 0$ uniformly in $u \in \cU$, Lipschitz constant is~$\lp (F) \le \ell$ uniformly in $x \in \real^n$, and dispersion is uniformly bounded on~$\sigma_x^{2}$ with respect to the $P$-weighted Frobenius norm, i.e.,~$\|\Sigma (x,u)\|^{2}_{\mathsf{F},P^{1/2}} \le \sigma_x^2$. 

Our objective is to estimate tracking errors~$x_t - x^\star (v_t)$ for parameter-dependent equilibrium~$x^\star(v_t)$, i.e.,~$F(x^\star(v_t), v_t) = \0_n$. There are several possible scenarios: 1) deterministic input~$u_t = \theta (t)$ and deterministic equilibrium curve~$x^\star (\theta (t))$, i.e., $u_t = v_t = \theta (t)$; 2) stochastic input~$u_t = \theta (t) + \xi_t$ and deterministic equilibrium curve~$x^\star (\theta (t))$, i.e., $v_t = \theta (t)$; 3)  stochastic input~$u_t = \theta (t) + \xi_t$ and stochastic equilibrium curve~$x^\star (\theta (t)+ \xi_t)$, i.e., $u_t = v_t = \theta (t) + \xi_t$. In this subsection, we assume that~$\xi_t$ is driven by an OU process. Namely,
\begin{align}\label{eq2:SDE_OU_SI}
    u_t = \theta (t) + \xi_t, 
    \quad
    d \xi_t = - c \xi_t dt + \frac{\sigma_\xi}{\sqrt{m}}  d\bs_{t}^\xi.
\end{align}

We summarize the three main results of this subsection, where roles of parameter~$\alpha$ and constant~$h_{\mathsf{OU}}$, named the It\^o drift correction constant associated with the OU generator, are explained in Remark~\ref{rem:OU} below:
\begin{enumerate}
\item (Theorem~\ref{thm:trck_OU_DIDC}) deterministic input $u_t =\theta(t)$,
  tracking a deterministic curve $\|x_t-x^\star (\theta(t))\|$:
          \begin{align}\label{eq2:trck_OU_DIDC}
             &\limsup_{t \to \infty} \E \bigl[ \|x_{t} - x^\star(\theta(t))\|^{2}_{2,P^{1/2}} \bigr]
            \leq \frac{1}{\alpha} \frac{\sigma_x^{2}}{2c}
             +\frac{1}{4\alpha(1-\alpha)} \frac{\ell^2}{c^4}  \limsup_{t \to \infty} \|\dot  \theta (t)\|_{\cU}^{2}.
        \end{align}
        
\item (Theorem~\ref{thm:trck_OU_SIDC}) stochastic input $u_t =\theta(t)+\xi_t$, tracking a deterministic curve
  $\|x_t-x^\star(\theta(t))\|$:
        \begin{align}\label{eq2:trck_OU_SIDC}
             &\limsup_{t \to \infty}\E [ \| x_t - x^\star (\theta(t)) \|_{2,P^{1/2}}^{2}]
            \leq  \frac{1}{\alpha}  \frac{\sigma_x^2}{c} 
            +\frac{1}{\alpha (1-\alpha)} \frac{\ell^2}{c^4} \limsup_{t \to \infty} \|\dot \theta (t)\|_2^{2} + \frac{1}{\alpha} \frac{\ell^2 }{c^2} \frac{\sigma_\xi^2}{c}.
        \end{align}
  
\item (Theorem~\ref{thm:trck_OU_SISC}) stochastic input $u_t =\theta(t)+\xi_t$, tracking a stochastic curve
  $\|x_t-x^\star(u_t)\|$:
      \begin{align}\label{eq2:trck_OU_SISC}
        &\limsup_{t \to \infty} \E [ \|  x_t -  x^\star (u_t) \|_{2,P}^{2} ]
            \nonumber\\
            %--------------------------------------------
            &\quad \leq 
            \frac{1}{\alpha} \frac{\sigma_x^{2}}{2c} 
            +\frac{1}{\alpha(1-\alpha)} \frac{\ell^2}{c^4} \limsup_{t \to \infty} \|\dot \theta (t)\|_2^2
             + \frac{1}{\alpha (1-\alpha)}  \left( (2-\alpha) \frac{\ell^2}{c^2} \frac{\sigma_\xi^2}{2c} + \frac{h_{\mathsf{OU}}^2}{2} \frac{\sigma_\xi^4}{4c^2} \right),
    \end{align}
    where
    \begin{align}\label{eq:Hess_OU_SISC}
        h_{\mathsf{OU}} : = \frac{1}{m}\sup_{u \in \real^m}\left\|\begin{bmatrix}
             \operatorname{trace} \bigl(\operatorname{Hess} (x_1^\star (u))\bigr) 
             \\ \vdots \\
             \operatorname{trace} \bigl(\operatorname{Hess} (x_n^\star (u))\bigr) 
        \end{bmatrix}\right\|_{2,P^{1/2}}.
    \end{align}
\end{enumerate}

\begin{remark}[Comparisons]\label{rem:OU}
\begin{enumerate}
  \renewcommand{\theenumi}{\roman{enumi}}

  \item  Theorem~\ref{thm:trck_OU_DIDC} can be regarded as a generalization of Proposition~\ref{prop:eqt}, extending deterministic equilibrium tracking to the stochastic setting. As noted previously for Theorem~\ref{thm:NISS} on NISS, in the stochastic case, we use the squared norm to evaluate the second moment. The additional parameter $\alpha \in (0,1)$ arises from applying Young’s inequality to estimate an upper bound in the presence of both noise and deterministic input. Compared with~\eqref{eq:eqt_dtr} for the deterministic equilibrium tracking, the tracking error bound~\eqref{eq:trck_OU_DIDC} contains additional term $\frac{1}{\alpha}\frac{\sigma_x^{2}}{2c}$ caused by noise. As aforementioned, this term corresponds to the expected squared norm of the OU process~\eqref{eq:SDE_OU_var}.

  \item Theorems~\ref{thm:trck_OU_SIDC} and~\ref{thm:trck_OU_SISC} can be regarded as two different generalizations of Theorem~\ref{thm:trck_OU_DIDC}. In both cases, the additional term $\frac{\ell^2}{c^2}\frac{\sigma_\xi^2}{c}$ arises from the stochastic input~$\xi_t$. The coefficient $\frac{\ell^2}{c^2}$ is standard in ISS estimates. In addition, by means of a second-moment analysis, we obtain the bounds for the case in which~$\|\cdot\|_{\mathcal U}$ is the $\ell_2$-norm.

  \item In Theorem~\ref{thm:trck_OU_SISC}, there is a further additional constant~$h_{\mathsf{OU}}$. This is called the It\^o drift correction constant associated with the OU generator because it corresponds to the It\^o drift correction associated with the OU generator that appears when computing the SDE satisfied by stochastic equilibrium~$x^\star (u_t)$: 
  \begin{align}\label{pf3:trck_OU_SISC}
    dx^\star (u_t) &=F(x^\star (u_t), u_t) dt 
    +  v(\theta(t), \xi_t, u_t) dt + \Lambda (u_t) d\bs_{t}^\xi,
\end{align}
where
\begin{subequations}\label{pf4:trck_OU_SISC}
\begin{align}
    v(\theta(t), \xi_t, u_t)
    %--------------------------------------------
    &= \frac{\partial x^\star_k}{\partial u} (u_t) (\dot \theta(t) - c \xi_t ) 
    + \frac{1}{2} \frac{\sigma_\xi^2}{m} \begin{bmatrix}
             \operatorname{trace} \bigl(\operatorname{Hess} (x_1^\star(u_t))\bigr) 
             \\ \vdots \\
             \operatorname{trace} \bigl(\operatorname{Hess} (x_n^\star(u_t))\bigr) 
        \end{bmatrix}
    \label{pf41:trck_OU_SISC}\\
    %--------------------------------------------
    %--------------------------------------------
    \Lambda (u_t) 
    %--------------------------------------------
    &= \frac{\sigma_\xi}{\sqrt{m}}\frac{\partial x^\star}{\partial u} (u_t).
    \label{pf42:trck_OU_SISC}
\end{align}
\end{subequations}
This SDE is obtained by the It\^o formula to~$x^\star (u_t)$. For more details, see Appendix~\ref{app:trck_OU_SISC}. 
  \red
\end{enumerate}
\end{remark}

%%%%%%%%%%%%%%%%%%%%%%%%%%%%%%%%%%%%%%%%%%%%%%%%%%%%%%%%%%%%%%%%%%
%%%%%%%%%%%%%%%%%%%%%%%%%%%%%%%%%%%%%%%%%%%%%%%%%%%%%%%%%%%%%%%%%%

\subsubsection{Deterministic Input, Tracking a Deterministic Curve}
We first focus on deterministic parameter-dependent equilibrium~$x^\star(\theta)$, i.e.,~$F(x^\star(\theta), \theta) = \0_n$ and then estimate the tracking error~$x_t - x^\star(\theta (t))$ for the solution to the following SDE:
\begin{align}\label{eq:SDE_OU_DIDC}
        dx_{t}
        =F(x_{t}, \theta (t)) dt 
        + \Sigma (x_{t}, \theta (t)) d\bs_{t}^x.
\end{align}

\begin{theorem}[Stochastic Equilibrium Tracking: Deterministic Input, Tracking a Deterministic Curve]
  \label{thm:trck_OU_DIDC}
    Given an input-dependent vector field $F:\real^{n} \times \cU \to \real^{n}$, and matrix $\Sigma : \real^{n} \times \real^{m} \to \real^{n \times r}$,  where $\cU \subset \real^m$ is compact, we impose the same assumptions as Theorem~\ref{thm:NISS} (i.e., i)~Assumption~\ref{asm:lip} of Lipschitz continuity and linear growth for the drift vector field $F$ and the dispersion matrix $\Sigma$; and ii)~the existence of a matrix $P=P^{\top}\succ \0_{n \times n}$ satisfying Assumptions~\ref{A1:NISS} to~\ref{A3:NISS} for contractivity of~$F$, Lipschitzness  of~$F$, and the boundedness of the dispersion of~$\Sigma$).
    
    Let~$x^\star(\theta(t))$ be a time-varying equilibrium curve, i.e., $F(x^\star(\theta(t)), \theta(t)) \equiv \0_n$.
    Then, for each $\alpha\in(0,1)$ and any finite $t \ge 0$, for any realization $x_t$ of the SDE~\eqref{eq:SDE_OU_DIDC} from random square integrable initial condition $x_0$ under continuously differentiable deterministic parameter $\theta:\real\to\cU$, we have
        \begin{align}\label{eq:trck_OU_DIDC}
            \E \bigl[ \|x_{t} - x^\star(\theta(t))\|^{2}_{2,P^{1/2}} \bigr]
            %--------------------------------------------
            &\leq \E \bigl[ \|x_{0}-x^\star(\theta(0))\|^{2}_{2,P^{1/2}} \bigr]\e^{-2c\alpha t}
            +\frac{1}{\alpha}\frac{\sigma_x^{2}}{2c}(1-\e^{-2c\alpha t}) \nonumber\\
            &\quad
             +\frac{1}{1-\alpha}\frac{\ell^2}{2c^3}\int_{0}^{t}\e^{-2c\alpha(t-\tau)}\|\dot  \theta (t)\|_{\cU}^{2}\ d\tau,
        \end{align}
        and~\eqref{eq2:trck_OU_DIDC}.
\end{theorem}

\begin{proof}
Consider an auxiliary dynamics of the SDE~\eqref{eq:SDE_OU_DIDC}:
    \begin{align}\label{pf1:trck_OU_DIDC}
         dx_{t} = F(x_{t}, \theta (t) ) dt +  v(t) dt  + \Sigma (x_{t}, \theta (t)) d\bs_{t}^x.
    \end{align}
    When $v (t) \equiv \0_n$, this is nothing but~\eqref{eq:SDE_OU_DIDC}. When $v(t) \equiv \dot x^\star (\theta (t))$ and noise free (i.e., $d\bs_{t}^x$ is identical to zero), we have~$x_{t}=x^\star (\theta (t))$. Namely, $x^\star (\theta (t))$ is a solution to \eqref{pf1:trck_OU_DIDC}.
    
    Applying item~\ref{itm2:NISS} of Theorem~\ref{thm:NISS} to the auxiliary SDE~\eqref{pf1:trck_OU_DIDC} as~$v(t)$ as the input, we have
    \begin{align*}
        \E \bigl[ \|x_{t} - x^\star(\theta(t))\|^{2}_{2,P^{1/2}} \bigr]
         \leq& \E \bigl[ \|x_{0}-x^\star(\theta(0))\|^{2}_{2,P^{1/2}} \bigr]\e^{-2c\alpha t}+\frac{\sigma_x^{2}}{2c\alpha}(1-\e^{-2c\alpha t})
        \nonumber\\
        & +\frac{1}{2c(1-\alpha)}\int_{0}^{t}\e^{-2c\alpha(t-\tau)}\|\dot x^\star (\theta (t))\|_{2,P^{1/2}}^{2}\ d\tau,
    \end{align*}
    From $\|\dot x^\star (\theta (\tau))\|_{2,P^{1/2}} \le \frac{\ell}{c} \|\dot \theta (\tau)\|_\cU$, we have~\eqref{eq:trck_OU_DIDC}.
    Finally, \eqref{eq2:trck_OU_DIDC} is obtained by taking the limit superior of~\eqref{eq:trck_OU_DIDC}.
\end{proof}

%%%%%%%%%%%%%%%%%%%%%%%%%%%%%%%%%%%%%%%%%%%%%%%%%%%%%%%%%%%%%%%%%%
%%%%%%%%%%%%%%%%%%%%%%%%%%%%%%%%%%%%%%%%%%%%%%%%%%%%%%%%%%%%%%%%%%

\subsubsection{Stochastic Input, Tracking a Deterministic Curve}

\begin{theorem}[Stochastic Equilibrium Tracking with OU process: Stochastic Input, Tracking a Deterministic Curve]\label{thm:trck_OU_SIDC}
    Given~$F:\real^{n} \times \cU \to \real^{n}$ and~$\Sigma : \real^{n} \times \real^{m} \to \real^{n \times r}$,  where $\cU \subset \real^m$ is compact, we impose the same assumptions as Theorem~\ref{thm:NISS}, where~$\| \cdot \|_\cU = \| \cdot \|_2$. In addition, we assume that
    \begin{enumerate}
    \setcounter{enumi}{3}
    \renewcommand{\theenumi}{A\arabic{enumi}}
        \item \label{A4:trck_OU_SIDC} 
        Normalization of $P$: $\|P\|_2 = 1$.
    \end{enumerate}
    Let~$x^\star(\theta(t))$ be a time-varying equilibrium curve.
    Then, for each $\alpha\in(0,1)$ and any finite $t \ge 0$, for any realization $(x_t,\xi_t)$ of the SDEs~\eqref{eq1:SDE_OU_SI} and~\eqref{eq2:SDE_OU_SI} from random square integrable initial condition $(x_0,\xi_0)$ under continuously differentiable deterministic parameter $\theta:\real\to\cU$, we have
     \begin{align}\label{eq:trck_OU_SIDC}
              \E [ \| x_t - x^\star(\theta(t)) \|_{2,P^{1/2}}^{2}]
            \le& 
            \E [ \| x_0 - x^\star(\theta(0)) \|_{2,P^{1/2}}^{2} ] \e^{-c \alpha t}
            + \frac{1}{\alpha} \frac{\sigma_x^2}{c} (1-\e^{-c \alpha t})
            \nonumber\\
            &
              +\frac{1}{1-\alpha} \frac{\ell^2}{c^3}\int_{0}^{t}\e^{-c\alpha(t-\tau)}\|\dot \theta (\tau)\|_2^{2}\ d\tau
            + \frac{\ell^2}{c^2}\E [ \|  \xi_0 \|_2^2 ]  \e^{-c \alpha t}
             + \frac{1}{\alpha} \frac{\ell^2 }{c^2} \frac{\sigma_\xi^2}{c} (1-\e^{-c \alpha t}),
        \end{align}
        and~\eqref{eq2:trck_OU_SIDC}.
\end{theorem}
\begin{proof}
    The proof is in Appendix~\ref{app:trck_OU_SIDC}.
\end{proof}

%%%%%%%%%%%%%%%%%%%%%%%%%%%%%%%%%%%%%%%%%%%%%%%%%%%%%%%%%%%%%%%%%%
%%%%%%%%%%%%%%%%%%%%%%%%%%%%%%%%%%%%%%%%%%%%%%%%%%%%%%%%%%%%%%%%%%

\subsubsection{Stochastic Input, Tracking a Stochastic Curve}

\begin{theorem}[Stochastic Equilibrium Tracking  with OU process: Stochastic Input, Tracking a Stochastic Curve]\label{thm:trck_OU_SISC}
 Given~$F:\real^{n} \times \cU \to \real^{n}$, and~$\Sigma : \real^{n} \times \real^{m} \to \real^{n \times r}$, where $\cU \subset \real^m$ is compact, we impose the same assumptions as Theorem~\ref{thm:trck_OU_DIDC}. Let~$x^\star(u_t)$ be a stochastic equilibrium curve, where~$u_t$ is generated by~\eqref{eq2:SDE_OU_SI}. 
 In addition, we assume that
    \begin{enumerate}
    \setcounter{enumi}{4}
    \renewcommand{\theenumi}{A\arabic{enumi}}    
    \item \label{A5:trck_OU_SISC}  Twice continuous differentiability of $F$: $F:\real^n \times \cU \to \real^n$ is of class~$C^2$;

    \item \label{A6:trck_OU_SISC}  Boundedness and Lipshitzness of the Hessian of $x_i^\star (v)$: $\operatorname{Hess} (x_i^\star (v))$ is bounded and Lipschitz on $\real^m$ for all $i=1,\dots,m$.
    \end{enumerate}
    Then, for each $\alpha\in(0,1)$ and any finite $t \ge 0$, for any weak solution $(\xi_t, x_t)$ of SDE~\eqref{eq1:SDE_OU_SI} and~\eqref{eq2:SDE_OU_SI} from random square integrable initial conditions $(x_0, \xi_0)$ under continuously differentiable deterministic parameter $\theta:\real\to\cU$, we have
    \begin{align}\label{eq:trck_OU_SISC}
      \E \bigl[ \|x_{t} - x^\star (u_t) \|^{2}_{2,P^{1/2}} \bigr]
            &\leq \E \bigl[ \|x_{0}-x^\star (u_0) \|^{2}_{2,P^{1/2}} \bigr]\e^{-2c\alpha t}
            +\frac{1}{\alpha} \frac{\sigma_x^{2}}{2c} (1-\e^{-2c\alpha t})
            \nonumber\\
            &\quad +\frac{2}{1-\alpha} \frac{\ell^2}{c^3}\int_{0}^{t}\e^{-2c\alpha(t-\tau)}  \| \dot \theta(\tau) \|_2^2\ d\tau
             +\frac{1}{(1-\alpha)^2} \frac{\ell^2}{2c^2}  \E \bigl[ \| \xi_0 \|_2^2  \bigr] (e^{-2c\alpha t}-e^{-2ct})
            \nonumber\\
            &\quad + \frac{1}{\alpha}  \left( \frac{\ell^2}{c^2} \frac{\sigma_\xi^2}{2c} +\frac{1}{1-\alpha} \frac{h_{\mathsf{OU}}^2}{2} \frac{\sigma_\xi^4}{4c^2} \right) (1-\e^{-2c\alpha t})  
            \nonumber\\
            &\quad 
            +\frac{1}{1-\alpha} \frac{\ell^2}{c^2} \frac{\sigma_\xi^2}{2c} \left(\frac{1}{\alpha} - \frac{1}{\alpha (1-\alpha)} e^{-2c\alpha t}+ \frac{1}{1-\alpha} \e^{-2ct}\right),
    \end{align}
    and~\eqref{eq2:trck_OU_SISC}.
\end{theorem}
\begin{proof}
    The proof is in Appendix~\ref{app:trck_OU_SISC}.
\end{proof}

%%%%%%%%%%%%%%%%%%%%%%%%%%%%%%%%%%%%%%%%%%%%%%%%%%%%%%%%%%%%%%%%%%
%%%%%%%%%%%%%%%%%%%%%%%%%%%%%%%%%%%%%%%%%%%%%%%%%%%%%%%%%%%%%%%%%%

\subsection{Driven by Jacobi Diffusion Processes}

To deal with a case where stochastic input~$u_t$ is bounded, in this subsection we consider noise driven by a multivariate JD process:
\begin{align}\label{eq:SDE_JD_SI}
    d u_t &= - c (u_t - \theta (t) )dt  
    + \sigma_u \operatorname{diag}(u_t \odot (a -u_t))^\frac{1}{2} d\bs_{t}^u.
\end{align}

We summarize the two main results of this subsection for noise driven by a JD process, which are parallel to those obtained in the OU case. In each case, we consider the noisy contracting dynamics~\eqref{eq1:SDE_OU_SI} characterized by contraction rate $c$, Lipschitz constant $\ell$, and dispersion bound $\sigma_x^2$. The parameter~$\alpha$ emerges from the application of Young’s inequality to bound the combined effect of stochastic disturbances and deterministic inputs, as in the OU case. Also, the role of constant~$h_{\mathsf{JD}}$, named the It\^o drift correction constant associated with the JD generator, is simitar to that of $h_{\mathsf{OU}}$ for the OU process.

\begin{enumerate}
\item (Theorem~\ref{thm:trck_JD_SIDC}) stochastic input~$u_t$, tracking a deterministic curve
  $\|x_t-x^\star(\theta(t))\|$:
        \begin{align}\label{eq2:trck_JD_SIDC}
             & \limsup_{t \to \infty} \E [ \| x_t - x^\star(\theta(t)) \|_{2,P^{1/2}}^{2}]
            \leq 
             \frac{1}{\alpha} \frac{\sigma_x^2}{c}
            +\frac{1}{\alpha (1-\alpha)} \frac{\ell^2}{c^4} \limsup_{t \to \infty} \|\dot \theta (t)\|_2^{2} 
             + \frac{1}{\alpha} \frac{\ell^2 }{c^2} \frac{\|a\|_2^2}{4} \frac{\sigma_u^2}{c} .
        \end{align}
  
\item (Theorem~\ref{thm:trck_JD_SISC}) stochastic input~$u_t$, tracking a stochastic curve
  $\|x_t-x^\star(u_t)\|$:
          \begin{align}\label{eq2:trck_JD_SISC}
        &\limsup_{t \to \infty} \E \bigl[ \|x_{t} - x^\star (u_t) \|^{2}_{2,P^{1/2}} \bigr]
            \nonumber\\
            %--------------------------------------------
            &\quad 
            \le \frac{1}{\alpha} \frac{\sigma_x^{2}}{2c}
             +\frac{1}{2\alpha(1-\alpha)} \frac{\ell^2}{c^4} \limsup_{t \to \infty} \| \dot \theta(t) \|_2^2 
            + \frac{1}{\alpha (1-\alpha)} \left( (4-3\alpha) \frac{\ell^2}{c^2}\frac{ \|a\|_2^2}{4}  \frac{\sigma_u^2}{2c} + \frac{h_{\mathsf{JD}}^2}{2} \frac{\sigma_u^4}{4c^2}\right),
    \end{align}
    where
    \begin{align}\label{eq:Hess_JD_SISC}
        h_{\mathsf{JD}} :=  \sup_{\theta \in (\0_m, a)} \left\|\sum_{i=1}^m u_i (a_i - u_i) \frac{\partial^2 x^\star (u)}{\partial u_i^2}\right\|_{2,P^{1/2}}.
    \end{align}
\end{enumerate}
Although the JD process has a different structure form the OU process, the basic structure of the error bounds is similar. In the JD case, the SDE satisfied by stochastic equilibrium~$x^\star (u_t)$ is
  \begin{align*}
    dx^\star (u_t) &=F(x^\star (u_t), u_t) dt 
    +  v(\theta(t), u_t) dt + \Lambda (u_t) d\bs_{t}^\xi,
\end{align*}
where
\begin{subequations}\label{pf4:trck_JD_new}
\begin{align}
    v(\theta(t), u_t)
    %--------------------------------------------
    &= \frac{\partial x^\star}{\partial \theta} (u_t) (- c (u_t -\theta(t)) )
    + \frac{1}{2} \sigma_u^2 \sum_{i=1}^m u_{t,i} (a_i - u_{t,i}) \frac{\partial^2 x^\star (u_t)}{\partial u_{t,i}^2}
    \label{pf31:trck_JD_new}\\
    %--------------------------------------------
    %--------------------------------------------
    \Lambda (u_t) 
    %--------------------------------------------
    &= \sigma_u \frac{\partial x^\star}{\partial u} (u_t) \operatorname{diag}(u_t \odot (a - u_t))^\frac{1}{2}.
    \label{pf32:trck_JD_new}
\end{align}
\end{subequations}
Thus, $h_{\mathsf{JD}}$ is called the It\^o drift correction constant associated with the JD generator.

%%%%%%%%%%%%%%%%%%%%%%%%%%%%%%%%%%%%%%%%%%%%%%%%%%%%%%%%%%%%%%%%%%
%%%%%%%%%%%%%%%%%%%%%%%%%%%%%%%%%%%%%%%%%%%%%%%%%%%%%%%%%%%%%%%%%%

\subsubsection{Stochastic Input, Tracking a Deterministic Curve}

\begin{theorem}[Stochastic Equilibrium Tracking with respect to Jacobi Diffusion Equilibrium Curve: Stochastic Input, Tracking a Deterministic Curve]\label{thm:trck_JD_SIDC}
    Given~$F:\real^{n} \times \cU \to \real^{n}$, and~$\Sigma : \real^{n} \times \real^{m} \to \real^{n \times r}$, where $\cU \subset \real^m$ is compact, we impose the same assumptions as Theorem~\ref{thm:trck_OU_DIDC}. In addition, we assume that
 \begin{enumerate}
    \setcounter{enumi}{6}
    \renewcommand{\theenumi}{A\arabic{enumi}}
    \item \label{A7:trck_JD} Feller condition: for any~$t \ge 0$,
    \begin{align*}
    \frac{\sigma_u^2}{2c} a \le \theta (t) \le \left( 1 - \frac{\sigma_u^2}{2c}\right) a.
    \end{align*}
\end{enumerate}
    Let~$x^\star(\theta(t))$ be a time-varying equilibrium curve.
    Then, for each $\alpha\in(0,1)$ and any finite $t \ge 0$, for any realization $(x_t,u_t)$ of the SDE~\eqref{eq1:SDE_OU_SI} with~\eqref{eq:SDE_JD_SI} from random square integrable initial condition $(x_0,u_0)$ under continuously differentiable deterministic parameter $\theta:\real\to\cU$, we have
     \begin{align}\label{eq:trck_JD_SIDC}
             \E [ \| x_t - x^\star(\theta(t)) \|_{2,P^{1/2}}^{2}]
            \leq &
            \E [ \| x_0 - x^\star(\theta(0)) \|_{2,P^{1/2}}^{2} ] \e^{-c \alpha t}
            + \frac{1}{\alpha} \frac{\sigma_x^2}{c} (1-\e^{-c \alpha t})
            \nonumber\\
            &
              +\frac{1}{1-\alpha} \frac{\ell^2}{c^3}\int_{0}^{t}\e^{-c\alpha(t-\tau)}\|\dot \theta (\tau)\|_2^{2}\ d\tau
            \nonumber\\
            &
            + \frac{\ell^2}{c^2}\E [ \|  \xi_0 \|_2^2 ]  \e^{-c \alpha t}
             + \frac{1}{\alpha} \frac{\ell^2 }{c^2} \frac{\|a\|_2^2}{4} \frac{\sigma_u^2}{c} (1-\e^{-c \alpha t}),
        \end{align}
        and~\eqref{eq2:trck_JD_SIDC}.
\end{theorem}

\begin{proof}
    The proof is in Appendix~\ref{app:trck_JD_SIDC}.
\end{proof}

%%%%%%%%%%%%%%%%%%%%%%%%%%%%%%%%%%%%%%%%%%%%%%%%%%%%%%%%%%%%%%%%%%
%%%%%%%%%%%%%%%%%%%%%%%%%%%%%%%%%%%%%%%%%%%%%%%%%%%%%%%%%%%%%%%%%%

\subsubsection{Stochastic Input, Tracking a Stochastic Curve}

\begin{theorem}[Stochastic Equilibrium Tracking with respect to Jacobi Diffusion Equilibrium Curve: Stochastic Input, Tracking a Stochastic Curve]\label{thm:trck_JD_SISC}
 Given~$F:\real^{n} \times \cU \to \real^{n}$, and~$\Sigma : \real^{n} \times \real^{m} \to \real^{n \times r}$, where $\cU \subset \real^m$ is compact, we impose the same assumptions as Theorem~\ref{thm:trck_JD_SIDC} and Assumption~\ref{A5:trck_OU_SISC} of Theorem~\ref{thm:trck_OU_SISC}.
  Let~$x^\star(u_t)$ be a stochastic equilibrium curve, where~$u_t$ is generated by~\eqref{eq:SDE_JD_SI}.
    Then, for each $\alpha\in(0,1)$ and any finite $t \ge 0$, for any weak solution $(x_t,u_t)$ of the SDE~\eqref{eq1:SDE_OU_SI} with~\eqref{eq:SDE_JD_SI} from random square integrable initial condition $(x_0,u_0)$ under continuously differentiable deterministic parameter $\theta:\real\to\cU$, we have
    \begin{align}\label{eq:trck_JD_SISC}
        \E \bigl[ \|x_{t} - x^\star (u_t) \|^{2}_{2,P^{1/2}} \bigr]
             \leq& \E \bigl[ \|x_{0}-x^\star (u_0) \|^{2}_{2,P^{1/2}} \bigr]\e^{-2c\alpha t} 
            +\frac{1}{\alpha} \frac{\sigma_x^{2}}{2c} (1-\e^{-2c\alpha t})
            \nonumber\\
            &+\frac{1}{1-\alpha} \frac{\ell^2}{c^2} \int_{0}^{t}\e^{-2c\alpha(t-\tau)} \int_0^\tau \e^{-c (\tau-r)}  \| \dot \theta(r) \|_2^2 dr d\tau
            \nonumber\\
            & +\frac{1}{1-\alpha} \frac{\ell^2}{c} \E \bigl[ \| u_0 - \theta(0) \|_2^2\bigr]\int_{0}^{t}\e^{-2c\alpha(t-\tau)} \e^{-c \tau} \ d\tau
            \nonumber\\
            & + \frac{1}{\alpha} \left(  \frac{\ell^2}{c^2} \frac{3 \|a\|_2^2}{4} \frac{\sigma_u^2}{2c} +\frac{1}{1-\alpha} \frac{h_{\mathsf{JD}}^2}{2} \frac{\sigma_u^4}{4c^2}\right) ( 1 - e^{-2c\alpha t})
            \nonumber\\
            & +\frac{1}{1-\alpha} \frac{\ell^2}{c} \frac{\| a\|_2^2}{4} \frac{\sigma_u^2}{c}\int_{0}^{t}\e^{-2c\alpha(t-\tau)} (1 - \e^{-c\tau}) \ d\tau.
    \end{align}
    and, if~$\alpha \ge 1/2$,~\eqref{eq2:trck_JD_SISC}.
\end{theorem}
\begin{proof}
    The proof is in Appendix~\ref{app:trck_JD_SISC}.
\end{proof}

%%%%%%%%%%%%%%%%%%%%%%%%%%%%%%%%%%%%%%%%%%%%%%%%%%%%%%%%%%%%%%%%%%
%%%%%%%%%%%%%%%%%%%%%%%%%%%%%%%%%%%%%%%%%%%%%%%%%%%%%%%%%%%%%%%%%%

%%%%%%%%%%%%%%%%%%%%%%%%%%%%%%%%%%%%%%%%%%%%%%%%%%%%%%%%%%%%%%%%%%
%%%%%%%%%%%%%%%%%%%%%%%%%%%%%%%%%%%%%%%%%%%%%%%%%%%%%%%%%%%%%%%%%%

%%%%%%%%%%%%%%%%%%%%%%%%%%%%%%%%%%%%%%%%%%%%%%%%%%%%%%%%%%%%%%%%%%
%%%%%%%%%%%%%%%%%%%%%%%%%%%%%%%%%%%%%%%%%%%%%%%%%%%%%%%%%%%%%%%%%%

\section{Incremental input-to-state stability in Wasserstein space}
\label{sec:incremental-ISS-Wasserstein}
In this section, as a different ISS property of SDEs, we study ISS of probability densities in Wasserstein spaces. Similarly to the previous sections, we assume Lipschitz continuity and linear growth for the drift vector field and the dispersion matrix as well as the contractivity and Lipschitz properties of the drift vector field. However, we do not impose boundedness of dispersion

We consider a simpler SDE than~\eqref{eq:SDE}:
\begin{align}\label{eq:SDE_smp}
    dx_{t} = F(x_{t}, u(t)) dt + \hat \Sigma (t) d\bs_{t},
\end{align}
where $\hat \Sigma :\real \to \real^{n \times r}$ is time-dependent matrix.
The time evolution of the probability density $\upmu(t,x)$ of the process $x_t$ governed by the SDE~\eqref{eq:SDE_smp}, is given by the Fokker-Planck equation (or called the Kolmogorov forward equation):
\begin{align}\label{eq:FP}
    \frac{\partial \upmu (t,x)}{\partial t} 
    %--------------------------------------------
    &= - \sum_{i=1}^n \frac{\partial (\upmu (t,x) F_i(x,u(t)))}{\partial x_i}
     + \frac{1}{2} \operatorname{trace}\bigl(\hat \Sigma (t)^{\top} \operatorname{Hess}_x ( \upmu(t,x) ) \hat \Sigma (t)\bigr).
\end{align}

As a distance of two probability measures, we employ the Wasserstein distance.

\begin{definition}[Wasserstein metric]
    Given an arbitrary norm $\| \cdot \|$ and $p \in [1, \infty]$, the \emph{$p$-Wasserstein distance} between the probability measures $\upmu^x$ and $\upmu^y$ with finite $p$-moments on $\real^n$ is defined by
    \begin{align*}
        W_{p}(\upmu^x,\upmu^y)
        %--------------------------------------------
        &:=\inf_{\pi\in\Pi (\upmu^x,\upmu^y)} \bigl(\E_{\pi}[\|x-y\|^{p}]\bigr)^{\frac{1}{p}},
    \end{align*}
    where $\Pi (\upmu^x,\upmu^y)$ is the set of joint probability measures with marginals $\upmu^x$ and $\upmu^y$, that is $\upmu^x(x) = \int_{\real^n} \pi (x,y) dy$ and $\upmu^y(y) = \int_{\real^n} \pi (x,y) dx$.
    \red
\end{definition}

As the main result of this section, we study incremental input-to-state stability in a Wasserstein distance as follows.

\begin{theorem}[Incremental input-to-state stability in Wasserstein distance]\label{thm:iISS_Wss}
Given an input-dependent vector field $F:\real^{n} \times \cU \to \real^{n}$, time-dependent matrix $\hat \Sigma :\real \to \real^{n \times r}$, and deterministic measurable inputs $u^{x},u^{y}:\real\to\cU$, where $\cU \subset \real^m$ is compact, consider the realizations driven by the same Brownian motion:
    \begin{subequations}\label{eq:SDE_smp2}
    \begin{align}
        dx_{t}&=F(x_{t}, u^{x}(t)) dt + \hat \Sigma (t) d\bs_{t},\\
        %--------------------------------------------
        dy_{t}&=F(y_{t}, u^{y}(t)) dt + \hat \Sigma (t) d\bs_{t},
    \end{align}        
    \end{subequations}
    and the corresponding Fokker-Planck equations~\eqref{eq:FP}. 
    We impose Assumption~\ref{asm:lip} of Lipschitz continuity and linear growth for the drift vector field $F$ and the dispersion matrix~$\hat \Sigma$.
    We additionally assume that there exists a norm $\|\cdot\|$ such that
    \begin{enumerate}
    \renewcommand{\theenumi}{A\arabic{enumi}}
    
        \item \label{A1:ISS_Wdst} Contraction in $x$: with respect to the state $x$, the map $F$ is strongly infinitesimally contracting with rate $c>0$ with respect to norm $\|\cdot\|$, uniformly in $u \in \cU$;
        
        \item \label{A2:ISS_Wdst} Lipschitz in $u$: with respect to the input $u$, the map $F$ from normed space~$(\real^n,\|\cdot\|)$ to normed space~$(\cU, \|\cdot\|_\cU)$ is Lipschitz continuous with constant $\ell>0$, uniformly in $x \in \real^n$.
        
    \end{enumerate}
    Then, for each $p \in [1, \infty]$ and any finite $t \ge 0$, any pair of solutions to the Fokker-Planck equations $\upmu_t^x(x) = \upmu^x (t, x)$ and $\upmu_t^y(y) = \upmu^y (t, y)$ with the initial distributions $\upmu^x (0, x) = \upmu_0^x(x)$ and $\upmu^y (0, y) = \upmu_0^y(y)$ with finite $p$-moments satisfies
    \begin{align}\label{eq:ISS_Wdst}
        W_p(\upmu_t^x,\upmu_t^y)
        %--------------------------------------------
        \leq \e^{-ct}W_p(\upmu_0^x,\upmu_0^y) 
        + \ell \int_0^t \e^{-c(t-\uptau)}\|u^x(\uptau)-u^y(\uptau)\|_{\mathcal{U}}d\uptau.
    \end{align}
\end{theorem}

\begin{proof}
    Let $x_t$ and $y_t$ denote the solution to the SDE~\eqref{eq:SDE_smp2} with $x_0$ drawn from $\upmu^x_0$ and $y_0$ drawn from $\upmu^y_0$, respectively and further assume $x_t$ and $y_t$ are driven by the same realization of Brownian motion. Since $x_t$ and $y_t$ are driven by the same Brownian motion, the process $x_t-y_t$ is subject to no process noise. From the contractivity (with respect to $z$) and Lipschitzness (wit respect to $u$) assumptions (Assumptions~\ref{A1:ISS_Wdst} and~\ref{A2:ISS_Wdst}) on $F$, we obtain
    \begin{align}\label{pf1:thm:ISS_Wdst}
        \|x_t - y_t\| 
        %--------------------------------------------
        &\leq \e^{-ct}\|x_0 - y_0\|
         + \ell \int_0^t \e^{-c(t-\uptau)} \|u^x (\uptau) - u^y (\uptau)\|_{\mathcal{U}}d\uptau \; \mbox{ a.s.}
    \end{align}
    Let $\pi_0$ denote a joint probability distribution with marginals $\upmu^x_0$ and $\upmu^y_0$, and let $\pi_t$ be the corresponding joint distribution for the joint solution $(x_t, y_t)$. 
    
    We take the expectation $\E_{\pi_0}[ \cdot ]$
    of~\eqref{pf1:thm:ISS_Wdst}. Select $X = \|x_t -y_t\|$,
    $Y=\e^{-ct}\|x_0 - y_0\|$, and $b = \ell \int_0^t \e^{-c(t-\uptau)}
    \|u^x (\uptau) - u^y (\uptau)\|_{\mathcal{U}}d\uptau$. Since
    $\E_{\pi_{t_0}}[\|Y\|^p]^{1/p}$ is finite for any $p \in [1, \infty]$
    (where $\E[\|\cdot\|^p]^{1/p}$ is read as
    $\mathrm{ess\,sup}\,\|\cdot\|$ when $p=\infty$), applying
    Lemma~\ref{lem:Minkowski} to~\eqref{pf1:thm:ISS_Wdst} yields
    \begin{align*}
        W_{p}(\upmu^x,\upmu^y)
        %--------------------------------------------
        &\le \E_{\pi_t} \bigl[ \|x_t - y_t\|^p\bigr ]^{1/p}
        \nonumber\\
        %--------------------------------------------
        &=\E_{\pi_0} \bigl[ \|x_t - y_t\|^p\bigr ]^{1/p}
        \leq \e^{-ct} \E_{\pi_0}\bigl[ \|x_0 - y_0\|^p\bigr ]^{1/p}
        + \ell \int_0^t \e^{-c(t-\uptau)} \|u^x (\uptau) - u^y (\uptau)\|_{\mathcal{U}}d\uptau.
    \end{align*}
    Taking the infimum with respect to $\pi_0\in\Pi (\upmu_0^x,\upmu_0^y)$ leads to~\eqref{eq:ISS_Wdst} for any $p \in [1, \infty]$.
\end{proof}

Let $(X,d)$ be a complete, separable metric space. For any finite~$p \ge
1$, the space $(\mathcal{P}_{p}(X),W_{p})$ is
complete~\citep{FB:08}. Moreover, when $p=+\infty$, also
$(\mathcal{P}_{\infty}^{b}(X), W_{\infty})$ is known to be
complete~\citep{CRG-RMS:84}, with $\mathcal{P}_{\infty}^{b}(X)$ being the
space of probability measures with bounded support.  Thus, if~$u^x(t)$ and
$u^y(t)$ are identical and constant, and~$\hat \Sigma$ is constant, then an
equilibrium distribution of the Fokker-Planck equation is globally
exponentially stable for any finite~$p \ge 1$.

\begin{corollary}[Global Exponential Stability of Equilibrium Distribution]
  Given~$F:\real^{n} \times \cU \to \real^{n}$ and constant~$\hat \Sigma
  \in \real^{n \times r}$, consider the pair of SDEs~\eqref{eq:SDE_smp2},
  where~$u^x(t) \equiv u^y(t) \equiv \bar u$ for some constant~$\bar u \in
  \cU$. If the same assumptions as Theorem~\ref{thm:iISS_Wss} hold, then
    \begin{enumerate}
        \item\label{item1:GES_Wss} for each~$p \in [1, \infty]$ and any finite $t \ge 0$, any pair of solutions to the Fokker-Planck equations $\upmu_t^x(x) = \upmu^x (t, x)$ and $\upmu_t^y(y) = \upmu^y (t, y)$ with the initial distributions $\upmu^x (0, x) = \upmu_0^x(x)$ and $\upmu^y (0, y) = \upmu_0^y(y)$ with finite $p$-moments satisfies
    \begin{align*}
        W_p(\upmu_t^x,\upmu_t^y)
        %--------------------------------------------
        \leq \e^{-ct}W_p(\upmu_0^x,\upmu_0^y);
    \end{align*}

        \item\label{item2:GES_Wss} for any~$p \in [1, \infty)$, an equilibrium distribution~$\mu^*$ for the Fokker-Planck equation~\eqref{eq:FP} is unique, and globally exponentially stable;

        \item\label{item3:GES_Wss} moreover, if~$F(x, \bar u) = -\nabla f(x)$ for a continuously differentiable~$c$-strongly convex function~$f:\real^n \to \real^n$, and~$\hat \Sigma = \sigma \I_n$, then the equilibrium distribution to which the Fokker-Planck equation converges is the Gibbs distribution with energy~$f$ and temperature~$\sigma^2/2$:
        \begin{align*}
            \mu^* (x) \propto \e^{-2 f(x)/\sigma^2}.
        \end{align*}
    \end{enumerate}
\end{corollary}
\begin{proof}
    Item~\ref{item1:GES_Wss} is a special case of Theorem~\ref{thm:iISS_Wss}. Item~\ref{item2:GES_Wss} follows from Banach contraction theorem~\cite[Theorem 1.6]{FB:26-CTDS}. 

    We show item~\ref{item3:GES_Wss}, i.e.,
    \begin{align*}
        \sum_{i=1}^n \frac{\partial (\upmu^* (x) F_i(x,\bar u))}{\partial x_i}
        =
        \frac{\sigma^2}{2} \sum_{i=1}^n \frac{\partial^2 \mu^*(x)}{\partial x_i^2},
    \end{align*}
    where~$F_i(x, \bar u) = - \frac{\partial f(x)}{\partial x_i}$ and~$\mu^* (x) \propto \e^{-2 f(x)/\sigma^2}$. Since the constant scaling coefficient can be canceled, it suffices to show
    \begin{align*}
        - \frac{\partial}{\partial x_i} \left( \frac{\partial f(x)}{\partial x_i} \e^{-2 f(x)/\sigma^2}\right)
        =
        \frac{\sigma^2}{2} \frac{\partial^2 \e^{-2 f(x)/\sigma^2}}{\partial x_i^2}.
    \end{align*}
    This can readily be shown by taking one round of derivatives in the right hand side.
\end{proof}

\section{Conclusion}
\label{sec:conclusions}
In this paper, we have developed contraction theory for SDEs driven by
deterministic inputs and stochastic noise.  Given a weighted~$\ell_2$-norm
for the state space, we have shown that the standard ISS conditions for
deterministic control systems imply NISS for the corresponding SDEs,
providing a natural generalization of conventional noise-to-state (NSS) and
ISS analysis. We have further applied our NISS analysis to estimate error
bounds for stochastic equilibrium tracking under different scenarios: 1)
deterministic input with a deterministic equilibrium curve, 2) stochastic
input with a deterministic equilibrium curve, and 3) stochastic input with
a stochastic equilibrium curve, considering two types of stochastic
processes, OU process and JD process, to represent unbounded and bounded
noise, respectively. Finally, we have studied contractivity of SDEs with
respect to a Wasserstein metric and shown that the standard ISS conditions
for deterministic control systems also imply ISS with respect to an
arbitrary~$p$-Wasserstein metric.

\appendix
\section{Proof of Theorem~\ref{thm:trck_OU_SIDC}}\label{app:trck_OU_SIDC}
Before estimating the tracking error $x_t - x^\star (\theta(t))$, we estimate the contractivity rate and Lipschitz constant of the cascade interconnection:
\begin{subequations}\label{eq:ODE_cscd}
     \begin{align}
         \dot \xi (t) &= - c \xi (t), 
         \\
         %--------------------------------------------
         \dot x(t) &= F(x(t), u(t)), \quad u(t) = \theta(t)+\xi (t).
     \end{align}
\end{subequations}

\begin{lemma}\label{lem:lgnrm_cscd}
    Given a vector field $F: \real^{n}\times \real^{m} \to\real^{n}$, assume that there exists a matrix $P=P^{\top}\succ \0_{n \times n}$ such that
    \begin{enumerate}
    \renewcommand{\theenumi}{A\arabic{enumi}}
        \item Normalization of $P$: $\|P\|_2 = 1$;
        
        \item Contraction in $x$: there exists $c>0$ such that $\olp (F) \le - c$, uniformly in $u \in \cU$;
        
        \item Lipschitz in $u$: for~$\| \cdot \|_\cU = \|\cdot\|_2$, there exists $\ell>0$ such that $\lp (F) \le \ell$, uniformly in $x \in \real^n$.
    \end{enumerate}
    Then, with respect to the state $(\xi, x)$, the cascade interconnection~\eqref{eq:ODE_cscd} is strongly infinitesimally contracting with rate $\frac{c}{2}>0$ in the weighted norm $\|\cdot\|_{2,P_\frac{c}{\ell}^{1/2}}$, $P_\frac{c}{\ell} := \begin{bmatrix}
    \frac{\ell}{c} \I_m & \0_{m \times n} \\ \0_{n \times m} & \frac{c}{\ell} P
    \end{bmatrix}$ uniformly in $u \in \real^m$.
\end{lemma}
\begin{proof}
    Define $P_\varepsilon := \begin{bmatrix}
    \varepsilon^{-1} \I_m & \0_{m \times n} \\ \0_{n \times m} & \varepsilon P
    \end{bmatrix}$, $\varepsilon > 0$. Similarly to~\cite[E2.40]{FB:26-CTDS}, we have
    \begin{align*}
       \mathsf{osLip}_{2, P_\varepsilon^{1/2}} \begin{bmatrix} - c \operatorname{id} \\ F\end{bmatrix}
        \le \max\{-c, \mathsf{osLip}_{2,P^{1/2}}^x (F)\} + \frac{\varepsilon}{2} \mathsf{Lip}_2^u (P^{1/2}F )
        %--------------------------------------------
        \le - c + \frac{\varepsilon \ell}{2} \|P^{1/2}\|_2 
        %--------------------------------------------
        = - c + \frac{\varepsilon \ell}{2}.
    \end{align*}
    The statement holds by selecting $\varepsilon = \frac{c}{\ell}$.
\end{proof}

\begin{proof}[Proof of Theorem~\ref{thm:trck_OU_SIDC}]
Consider an auxiliary dynamics of the SDE~\eqref{eq1:SDE_OU_SI} and~\eqref{eq2:SDE_OU_SI}:
    \begin{subequations}\label{pf1:trck_OU_SIDC}
    \begin{align}
        u_t &= \theta (t) + \xi_t,
        \quad d \xi_t = - c \xi_t dt + \frac{\sigma_\xi}{\sqrt{m}}  d\bs_{t}^\xi, \\
         dx_{t}&=F(x_{t}, u_t) dt +  v(t) dt  + \Sigma (x_{t}, u_t) d\bs_{t}^x.
    \end{align}
    \end{subequations}
    When $v(t) \equiv \0_n$, this is nothing but~\eqref{eq1:SDE_OU_SI} and~\eqref{eq2:SDE_OU_SI}. When $v(t) \equiv \dot x^\star (\theta (t))$ and noise free (i.e., $d\bs_{t}^\xi$ and $d\bs_{t}^x$ are both identically equal to zero and~$\xi_0 = \0_m$), we have~$x_{t}=x^\star (\theta (t))$. Namely, $x^\star (\theta (t))$ is a solution to \eqref{pf1:trck_OU_SIDC}.
    
    We consider applying item~\ref{itm2:NISS} of Theorem~\ref{thm:NISS} to the auxiliary SDE~\eqref{pf1:trck_OU_SIDC}. From Lemma~\ref{lem:lgnrm_cscd}, with respect to the state $(\xi, x)$, the vector field~$\begin{bmatrix}- c \xi \\ F(x, \theta + \xi) + v\end{bmatrix}$ is strongly infinitesimally contracting with rate $\frac{c}{2}>0$ in the weighted norm $\|\cdot\|_{2,P_\frac{c}{\ell}^{1/2}}$, uniformly in $\theta \in \cU$ and~$v\in \real^n$, and with respect to the input~$v$, the Lipschitz constant is~$1$. 
    Also, from the bounded dispersion assumption (Assumption~\ref{A3:NISS}), we have
        \begin{align*}
        &\operatorname{trace}\Biggl(\begin{bmatrix}
             \frac{\sigma_\xi}{\sqrt{m}} & \0_{m \times n}\\
            \0_{m \times n} & \Sigma (x, u)
        \end{bmatrix}^\top
        \begin{bmatrix}
                \frac{\ell}{c} \I_m & \0_{m \times n} \\ 
                \0_{n \times m} & \frac{c}{\ell} P
            \end{bmatrix}
        \begin{bmatrix}
             \frac{\sigma_\xi}{\sqrt{m}} & \0_{m \times n}\\
            \0_{m \times n} & \Sigma (x, u)
        \end{bmatrix}\Biggr)
        \nonumber\\
        %--------------------------------------------
        &\quad =
        \frac{\ell }{cm} \sigma_\xi^2 \operatorname{trace} \left( \I_m \right)
        +
        \frac{c}{\ell} \operatorname{trace} (\Sigma (x, u)^{\top}\ P\ \Sigma (x, u))
        \leq \frac{\ell }{c} \sigma_\xi^2 + \frac{c }{\ell} \sigma_x^2 .
    \end{align*}
    
    Applying item~\ref{itm2:NISS} of Theorem~\ref{thm:NISS} to the auxiliary SDE~\eqref{pf1:trck_OU_SIDC}, we have for~\eqref{eq1:SDE_OU_SI} and~\eqref{eq2:SDE_OU_SI} and~$x^\star (\theta(t))$,
    \begin{align}\label{pf3:trck_OU_SIDC}
             \E \biggl[ \left\| \begin{bmatrix} \xi_t \\ x_t \end{bmatrix} \! - \! \begin{bmatrix} \0_m \\ x^\star (\theta(t)) \end{bmatrix} \right\|_{2,P_\frac{c}{\ell}^{1/2}}^{2} \biggr]
            \leq& \E \biggl[ \left\| \begin{bmatrix} \xi_0 \\ x_0 \end{bmatrix} \! - \! \begin{bmatrix} \0_m \\ x^\star(\theta(0)) \end{bmatrix} \right\|_{2,P_\frac{c}{\ell}^{1/2}}^{2} \biggr] \e^{-c \alpha t}
            \nonumber\\
            & +\frac{1}{\alpha} \left(\frac{\ell }{c} \frac{\sigma_\xi^2}{c} + \frac{c }{\ell}\frac{\sigma_x^2 }{c}\right) (1-\e^{-c \alpha t})
            \nonumber\\
            & +\frac{1}{c (1-\alpha)}\int_{0}^{t}\e^{-c \alpha(t-\tau)}\|\dot x^\star (\theta (\tau))\|_{2,(\frac{c}{\ell}P)^{1/2}}^{2}\ d\tau.
    \end{align}
    From $\|\dot x^\star (\theta (\tau))\|_{2,P^{1/2}} \le \frac{\ell}{c} \|\dot \theta (\tau)\|_2$, we have
    \begin{align}\label{pf4:trck_OU_SIDC}
        \|\dot x^\star (\theta (\tau))\|_{2,(\frac{c}{\ell}P)^{1/2}}^{2}
        %--------------------------------------------
        \le \frac{\ell}{c}\|\dot \theta (\tau)\|_2^2.
    \end{align}
    Also, from the definition of~$P_{\frac{c}{\ell}}$, we obtain
    \begin{align}\label{pf5:trck_OU_SIDC}
    &\frac{c}{\ell} \E [ \| x_t - x^\star (\theta(t)) \|_{2,P^{1/2}}^{2}]
   \le 
    \E \biggl[ \left\| \begin{bmatrix} \xi_t \\ x_t \end{bmatrix} \! - \! \begin{bmatrix} \0_m \\ x^\star (\theta(t)) \end{bmatrix} \right\|_{2,P_\frac{c}{\ell}^{1/2}}^{2} \biggr].
    \end{align}
    Combining~\eqref{pf3:trck_OU_SIDC}--\eqref{pf5:trck_OU_SIDC} yields~\eqref{eq:trck_OU_SIDC}.
    Finally, \eqref{eq2:trck_OU_SIDC} is obtained by taking the limit superior of~\eqref{eq:trck_OU_SIDC}.
\end{proof}

%%%%%%%%%%%%%%%%%%%%%%%%%%%%%%%%%%%%%%%%%%%%%%%%%%%%%%%%%%%%%%%%%%
%%%%%%%%%%%%%%%%%%%%%%%%%%%%%%%%%%%%%%%%%%%%%%%%%%%%%%%%%%%%%%%%%%

%%%%%%%%%%%%%%%%%%%%%%%%%%%%%%%%%%%%%%%%%%%%%%%%%%%%%%%%%%%%%%%%%%
%%%%%%%%%%%%%%%%%%%%%%%%%%%%%%%%%%%%%%%%%%%%%%%%%%%%%%%%%%%%%%%%%%

\section{Proof of Theorem~\ref{thm:trck_OU_SISC}}\label{app:trck_OU_SISC}
\begin{proof}
Consider an auxiliary dynamics of the SDE~\eqref{eq1:SDE_OU_SI}:
\begin{align}\label{pf1:trck_OU_SISC}
    dx_{t}&=F(x_{t}, u_t) dt + \Sigma (x_{t}, u_t) d\bs_{t}^x +  v(\theta (t), \xi_t, u_t) dt + \Lambda (u_t) d\bs_{t}^\xi.
\end{align}
Note that $v(t) dt$ in~\eqref{pf1:trck_OU_DIDC} is replaced by $v(\theta (t), \xi_t, u_t) dt + \Lambda (u_t) d\bs_{t}^v$. If $v(\theta (t), \xi_t, u_t)$ and $\Lambda (u_t)$ are identically equal to zero, this is nothing but~\eqref{eq1:SDE_OU_SI}. When $d\bs_{t}^x$ is identical to zero, \eqref{pf1:trck_OU_SISC} becomes
\begin{align}\label{pf2:trck_OU_SISC}
    dx_{t}=F(x_{t}, u_t) dt +  v(\theta(t), \xi_t, u_t) dt + \Lambda (u_t) d\bs_{t}^\xi.
\end{align}
We first show that~$x^\star (u_t)$ satisfies this, i.e., \eqref{pf3:trck_OU_SISC} with~\eqref{pf4:trck_OU_SISC}. 

By the It\^o formula, the $k$th component of $x^\star (u_t)$ satisfies
\begin{align*}
    d x^\star_k (u_t)
    %--------------------------------------------
    &= \sum_{i=1}^m \frac{\partial x^\star_k}{\partial u_i} (u_t) d u_{i,t} 
    + \frac{1}{2} \sum_{i=1}^m \sum_{j=1}^m \frac{\partial^2 x^\star_k}{\partial u_i\partial u_j} (u_t) d u_{i,t} d u_{j,t},
\end{align*}
where $d\bs_{i,t}^\xi d\bs_{j,t}^\xi = \delta_{i,j}dt$ and $dt dt  = d\bs_{i,t}^\xi dt = dt d\bs_{i,t}^\xi = \0_m$. 

From~\eqref{eq2:SDE_OU_SI}, we have
\begin{align*}
    \sum_{i=1}^m \frac{\partial x^\star_k}{\partial u_i} (u_t) d u_{i,t}
    &= \frac{\partial x^\star_k}{\partial u} (u_t) d u_t
    = \frac{\partial x^\star_k}{\partial u} (u_t) \left((\dot \theta(t) - c \xi_t) dt + \frac{\sigma_\xi}{\sqrt{m}}  d\bs_{t}^\xi \right),
\end{align*}
and 
\begin{align*}
    &\sum_{i=1}^m \sum_{j=1}^m \frac{\partial^2 x^\star_k}{\partial u_i\partial u_j} (u_t) d u_{i,t} d u_{j,t}
    %--------------------------------------------
    =\frac{\sigma_\xi^2}{m} \operatorname{trace}\bigl(\operatorname{Hess} (x^\star_k (u_t))\bigr) dt.
\end{align*}
Thus, $x^\star(u_t)$ is a solution to~\eqref{pf2:trck_OU_SISC} for~\eqref{pf4:trck_OU_SISC}.

We consider applying item~\ref{itm2:NISS} of Theorem~\ref{thm:NISS} to the auxiliary dynamics~\eqref{pf1:trck_OU_SISC}. We need to concern~$v(\theta(t), \xi_t, u_t)$ and~$\Lambda (u_t)$. From boundedness and Lipschitzness of $\frac{\partial^2 x_i^\star}{\partial \theta^2}$, $i=1,\dots,m$ (Assumption~\ref{A6:trck_OU_SISC}), they satisfy the Lipschitz continuity and linear growth assumptions. 

Next, from the bounded dispersion assumption (Assumption~\ref{A3:NISS} in Theorem~\ref{thm:NISS}) and~\eqref{pf42:trck_OU_SISC}, we have
        \begin{align*}
        &\operatorname{trace}\left(\begin{bmatrix}
              \Sigma (x, u) & \0_{n \times m}\\
              \0_{n \times r} & \Lambda (u)
        \end{bmatrix}^\top\right.
        \begin{bmatrix}
                P & - P \\
                - P & P
            \end{bmatrix}
        \left.\begin{bmatrix}
             \Sigma (x, u) & \0_{n \times m}\\
              \0_{n \times r} & \Lambda (u)
        \end{bmatrix}\right)
        \nonumber\\
        %--------------------------------------------
        &\quad =
         \operatorname{trace} (\Sigma (x, u)^{\top} P \Sigma (x, u))
         + \operatorname{trace}  (\Lambda (u)^{\top} P \Lambda (u))
        \nonumber\\
        %--------------------------------------------
        &\quad \leq \sigma_x^2
        + \frac{\sigma_\xi^2}{ m} \operatorname{trace}  \left(\frac{\partial x^\star}{\partial v} (v)^{\top} P  \frac{\partial x^\star}{\partial v} (v)\right)
         \leq \sigma_x^2 + \frac{\ell^2}{c^2} \sigma_\xi^2 \|P\|_2 
        %--------------------------------------------
        = \sigma_x^2 + \frac{\ell^2}{c^2} \sigma_\xi^2 .
    \end{align*}
Thus, repeating a similar calculation as~\eqref{pf3:NISS} for the auxiliary SDE~\eqref{pf1:trck_OU_SISC}, we have for~\eqref{eq1:SDE_OU_SI} and~\eqref{pf3:trck_OU_SISC},
\begin{align*}
    \mathcal{L}\| x_t  -  x^\star (u_t) \|_{2,P^{1/2}}^{2} 
    \le& - 2c \| x_t -x^\star (u_t) \|_{2,P^{1/2}}^2 
    + \sigma_x^2 + \frac{\ell^2}{c^2} \sigma_\xi^2
     + 2 \| x_t - x^\star (u_t) \|_{2, P^{1/2}} \| v(\theta(t), \xi_t, u_t) \|_{2, P^{1/2}}
    \nonumber\\
    %--------------------------------------------
    \le& - 2c\alpha \| x_t - x^\star (u_t) \|_{2,P^{1/2}}^2 
    + \sigma_x^2 + \frac{\ell^2}{c^2} \sigma_\xi^2
    + \frac{1}{2c (1-\alpha)}  \| v(\theta(t), \xi_t, u_t) \|_{2, P^{1/2}} \; \mbox{ a.s.}
\end{align*}
for each $\alpha \in (0,1)$. Similarly to the proof of Theorem~\ref{thm:NISS}, we have
\begin{align}\label{pf5:trck_OU_SISC}
            \E \bigl[ \|x_{t} - x^\star (u_t) \|^{2}_{2,P^{1/2}} \bigr]
            &\leq \E \bigl[ \|x_{0}-x^\star (u_0) \|^{2}_{2,P^{1/2}} \bigr]\e^{-2c\alpha t}
            +\frac{1}{\alpha}\left(\frac{\sigma_x^{2}}{2c} + \frac{\ell^2}{c^2} \frac{\sigma_\xi^2}{2c}\right)(1-\e^{-2c\alpha t})
            \nonumber\\
            &\quad+\frac{1}{1-\alpha} \frac{1}{2c}\int_{0}^{t}\e^{-2c\alpha(t-\tau)} \E \bigl [\| v(\theta(\tau), \xi_\tau, u_\tau) \|_{2, P^{1/2}}^2 \bigr]\ d\tau.
        \end{align}
for each $\alpha \in (0,1)$.

It remains to estimate an upper bound on $\E \bigl [\| v(\theta(t), \xi_t, u_t) \|_{2, P^{1/2}}^2 \bigr]$. From~\eqref{eq:Hess_OU_SISC} and~\eqref{pf41:trck_OU_SISC} with Young's inequality, we have
\begin{align*}
    \| v(\theta(t), \xi_t, u_t) \|_{2,P^{1/2}}^{2}
    &\le 2 \left\|\frac{\partial x^\star_k}{\partial u} (u_t) (\dot \theta(t) - c \xi_t )  \right\|_{2,P^{1/2}}^{2} +\frac{\sigma_\xi^4}{2m^2} \left\|\begin{bmatrix}
             \operatorname{trace} \bigl(\operatorname{Hess} (x_1^\star (u_t))\bigr) 
             \\ \vdots \\
             \operatorname{trace} \bigl(\operatorname{Hess} (x_n^\star (u_t))\bigr) 
        \end{bmatrix}\right\|_{2,P^{1/2}}^{2}
    \nonumber\\
    %--------------------------------------------
     &\le 2 \frac{\ell^2}{c^2} \| \dot \theta(t) - c \xi_t \|_2^2 + \frac{h_{\mathsf{OU}}^2}{2} \sigma_\xi^4.
\end{align*}
Compute with Young's inequality,
\begin{align*}
    \| \dot \theta(t) - c \xi_t \|_2^2
    %--------------------------------------------
    \le  2 \| \dot \theta(t) \|_2^2  + 2 c^2 \| \xi_t \|_2^2.
\end{align*}
For the OU process, we have
\begin{align*}
    \E \bigl[ \| \xi_t\|_2^2  \bigr]  
    %--------------------------------------------
    = \e^{-2ct} \E \bigl[ \| \xi_0 \|_2^2  \bigr] 
    +\frac{\sigma_\xi^2}{2c} (1-\e^{-2ct}).
\end{align*}
In summary, we obtain
\begin{align}\label{pf8:trck_OU_SISC}
    \E [ \| v(\theta(t), \xi_t, u_t) \|_{2,P^{1/2}}^{2} ]
    &\le 2 \frac{\ell^2}{c^2} \E[\| \dot \theta(t) - c \xi_t \|_2^2]  + \frac{h_{\mathsf{OU}}^2}{2} \sigma_\xi^4.
    \nonumber\\
    %--------------------------------------------
     &= 4 \ell^2 \e^{-2ct} \E \bigl[ \| \xi_0 \|_2^2  \bigr]
    + 2 \ell^2 \frac{\sigma_\xi^2}{c} (1-\e^{-2ct}) 
    + 4 \frac{\ell^2}{c^2} \| \dot \theta(t) \|_2^2 + \frac{h_{\mathsf{OU}}^2}{2} \sigma_\xi^4.
\end{align}
Combining~\eqref{pf5:trck_OU_SISC} and~\eqref{pf8:trck_OU_SISC} yields
\begin{align*}
      \E \bigl[ \|x_{t} - x^\star (u_t) \|^{2}_{2,P^{1/2}} \bigr]
            &\leq \E \bigl[ \|x_{0}-x^\star (u_0) \|^{2}_{2,P^{1/2}} \bigr]\e^{-2c\alpha t}
             +\frac{1}{\alpha}\left(\frac{\sigma_x^{2}}{2c} + \frac{\ell^2}{c^2} \frac{\sigma_\xi^2}{2c}\right)(1-\e^{-2c\alpha t})
            \nonumber\\
            &\quad +\frac{2}{1-\alpha} \frac{\ell^2}{c}  \E \bigl[ \| \xi_0 \|_2^2  \bigr] \int_{0}^{t}\e^{-2c\alpha(t-\tau)}  \e^{-2c\tau}\ d\tau
            \nonumber\\
            &\quad +\frac{1}{1-\alpha} \frac{\ell^2}{c} \frac{\sigma_\xi^2}{c} \int_{0}^{t}\e^{-2c\alpha(t-\tau)} (1-\e^{-2c\tau})\ d\tau
            \nonumber\\
            &\quad +\frac{2}{1-\alpha} \frac{\ell^2}{c^3}\int_{0}^{t}\e^{-2c\alpha(t-\tau)}  \| \dot \theta(\tau) \|_2^2\ d\tau
             +\frac{1}{1-\alpha} \frac{h_{\mathsf{OU}}^2}{4} \frac{\sigma_\xi^4}{c}  \int_{0}^{t}\e^{-2c\alpha(t-\tau)} \ d\tau
\end{align*}   
Computing the time integrations, we have~\eqref{eq:trck_OU_SISC}. 
Finally,~\eqref{eq2:trck_OU_SISC} is obtained by taking the limit superior of~\eqref{eq:trck_OU_SISC}.
\end{proof}

%%%%%%%%%%%%%%%%%%%%%%%%%%%%%%%%%%%%%%%%%%%%%%%%%%%%%%%%%%%%%%%%%%
%%%%%%%%%%%%%%%%%%%%%%%%%%%%%%%%%%%%%%%%%%%%%%%%%%%%%%%%%%%%%%%%%%

%%%%%%%%%%%%%%%%%%%%%%%%%%%%%%%%%%%%%%%%%%%%%%%%%%%%%%%%%%%%%%%%%%
%%%%%%%%%%%%%%%%%%%%%%%%%%%%%%%%%%%%%%%%%%%%%%%%%%%%%%%%%%%%%%%%%%

\section{Proof of Theorem~\ref{thm:trck_JD_SIDC}}\label{app:trck_JD_SIDC}

\begin{proof}
Consider an auxiliary dynamics of the SDE~\eqref{eq1:SDE_OU_SI} with~\eqref{eq:SDE_JD_SI}:
    \begin{subequations}\label{pf1:trck_JD_SIDC}
    \begin{align}
        d u_t &= - c (u_t - \theta (t) )dt + \sigma_u \operatorname{diag}(u_t \odot (a -u_t))^\frac{1}{2} d\bs_{t}^u, \\
         dx_{t}&=F(x_{t}, u_t) dt +  v(t) dt  + \Sigma (x_{t}, u_t) d\bs_{t}^x.
    \end{align}
    \end{subequations}
    When $v(t) \equiv \0_n$, this is nothing but~\eqref{eq1:SDE_OU_SI} with~\eqref{eq:SDE_JD_SI}. When $v(t) \equiv \dot x^\star (\theta (t))$ and noise free (i.e., $d\bs_{t}^u$ and $d\bs_{t}^x$ are both identical to zero and~$u_t \equiv \theta (t)$), we have~$x_{t}=x^\star (\theta (t))$. Namely, $x^\star (\theta (t))$ is a solution to \eqref{pf1:trck_JD_SIDC}.

    Similarly to the proof of Theorem~\ref{thm:trck_JD_SIDC}, we consider applying item~\ref{itm2:NISS} of Theorem~\ref{thm:NISS} to the auxiliary SDE~\eqref{pf1:trck_JD_SIDC}. From the Feller condition (Assumption~\ref{A7:trck_JD}), $u$-dynamics has a weak solution staying in $(\0_m, a)$. Therefore, from the Lipschitz continuity and linear growth assumptions for the drift vector field $F$ and the dispersion matrix $\Sigma$ (Assumption~\ref{asm:lip}), the auxiliary dynamics~\eqref{pf1:trck_JD_SIDC} satisfies the linear growth condition on $(\0_m, a)$, and thus a weak solution exists for any $u_0 \in (\0_m, a)$. In particular, $u_t \in (\0_m, a)$ for any $t > 0$ almost surely. Moreover, we can apply the Dynkin's formula (Proposition~\ref{prop:Dynkin}).
    
    Since~$u_t \in (0, a)$, we have
    \begin{align*}
        \operatorname{trace} (\operatorname{diag}(u_t \odot (a -u_t))) \le \frac{\|a\|_2^2}{4}.
    \end{align*}
    Also, from the bounded dispersion assumption (Assumption~\ref{A3:NISS}), we have~\eqref{pf2:trck_JD_SIDC}.
        \begin{align}\label{pf2:trck_JD_SIDC}
        &\operatorname{trace}\Biggl(\begin{bmatrix}
             \sigma_u \operatorname{diag}(u_t \odot (a -u_t))^\frac{1}{2} & \0_{m \times n}\\
            \0_{m \times n} & \Sigma (x, u)
        \end{bmatrix}^\top
        \begin{bmatrix}
                \frac{\ell}{c} \I_m & \0_{m \times n} \\ 
                \0_{n \times m} & \frac{c}{\ell} P
            \end{bmatrix}
        \begin{bmatrix}
             \sigma_u \operatorname{diag}(u_t \odot (a -u_t))^\frac{1}{2} & \0_{m \times n}\\
            \0_{m \times n} & \Sigma (x, u)
        \end{bmatrix}\Biggr)
        \nonumber\\
        %--------------------------------------------
        &=
       \frac{\ell}{c}\sigma_u^2 \operatorname{trace} (\operatorname{diag}(u_t \odot (a -u_t)) ) +
        \frac{c}{\ell} \operatorname{trace} (\Sigma (x, u)^{\top}\ P\ \Sigma (x, u))
        \leq \frac{\ell }{c} \frac{\|a\|_2^2}{4} \sigma_u^2 
         + \frac{c }{\ell} \sigma_x^2 .
    \end{align}
    
    Applying item~\ref{itm2:NISS} of Theorem~\ref{thm:NISS} to the auxiliary SDE~\eqref{pf1:trck_JD_SIDC}, we have for~\eqref{eq1:SDE_OU_SI} with~\eqref{eq:SDE_JD_SI} and~$x^\star (\theta(t))$,
    \begin{align}\label{pf3:trck_JD_SIDC}
             \E \biggl[ \left\| \begin{bmatrix} u_t \\ x_t \end{bmatrix} \! - \! \begin{bmatrix} \theta(t) \\ x^\star (\theta(t)) \end{bmatrix} \right\|_{2,P_\frac{c}{\ell}^{1/2}}^{2} \biggr]
            &\leq \E \biggl[ \left\| \begin{bmatrix} u_0 \\ x_0 \end{bmatrix} \! - \! \begin{bmatrix} \theta(0) \\ x^\star(\theta(0)) \end{bmatrix} \right\|_{2,P_\frac{c}{\ell}^{1/2}}^{2} \biggr] \e^{-c \alpha t}
            \nonumber\\
            &\quad +\frac{1}{\alpha} \left(\frac{\ell }{c} \frac{\|a\|_2^2}{4} \sigma_u^2 
         + \frac{c }{\ell} \sigma_x^2\right) (1-\e^{-c \alpha t})
            \nonumber\\
            &\quad +\frac{1}{c (1-\alpha)}\int_{0}^{t}\e^{-c \alpha(t-\tau)}\|\dot x^\star (\theta (\tau))\|_{2,(\frac{c}{\ell}P)^{1/2}}^{2}\ d\tau.
    \end{align}
    Combining~\eqref{pf4:trck_OU_SIDC}, \eqref{pf5:trck_OU_SIDC}, and~\eqref{pf3:trck_JD_SIDC} yields~\eqref{eq:trck_JD_SIDC}.
    Finally, \eqref{eq2:trck_JD_SIDC} is obtained by taking the limit superior of~\eqref{eq:trck_JD_SIDC}.
\end{proof}

%%%%%%%%%%%%%%%%%%%%%%%%%%%%%%%%%%%%%%%%%%%%%%%%%%%%%%%%%%%%%%%%%%
%%%%%%%%%%%%%%%%%%%%%%%%%%%%%%%%%%%%%%%%%%%%%%%%%%%%%%%%%%%%%%%%%%

%%%%%%%%%%%%%%%%%%%%%%%%%%%%%%%%%%%%%%%%%%%%%%%%%%%%%%%%%%%%%%%%%%
%%%%%%%%%%%%%%%%%%%%%%%%%%%%%%%%%%%%%%%%%%%%%%%%%%%%%%%%%%%%%%%%%%

\section{Proof of Theorem~\ref{thm:trck_JD_SISC}}\label{app:trck_JD_SISC}

\begin{proof}
For an auxiliary dynamics~\eqref{pf1:trck_OU_SISC} of the SDE~\eqref{eq1:SDE_OU_SI}, in the Jacobi diffusion case, we have~\eqref{pf4:trck_JD_new}. We consider applying item~\ref{itm2:NISS} of Theorem~\ref{thm:NISS} to the auxiliary dynamics~\eqref{pf1:trck_OU_SISC}. From the Feller condition (Assumption~\ref{A7:trck_JD}), $u$-dynamics has a weak solution staying in $(\0_m, a)$. Also, for continuity of $\frac{\partial^2 x_i^\star}{\partial u^2}$, $i=1,\dots,m$ (Assumption~\ref{A5:trck_OU_SISC} ), $v$ and $\Lambda$ in~\eqref{pf4:trck_JD_new} are bounded (when $u \in (\0_m, a)$). Therefore, from the Lipschitz continuity and linear growth assumptions for the drift vector field $F$ and the dispersion matrix $\Sigma$ (Assumption~\ref{asm:lip}), the auxiliary dynamics~\eqref{pf1:trck_OU_SISC} satisfies the linear growth condition on $(\0_m, a)$, and thus a weak solution exists for any $u_0 \in (\0_m, a)$. In particular, $u_t \in (\0_m, a)$ for any $t > 0$ almost surely. Moreover, we can apply the Dynkin's formula (Proposition~\ref{prop:Dynkin}).

Next, from the bounded dispersion assumption (Assumption~\ref{A3:NISS}) and~\eqref{pf32:trck_JD_new}, we have
        \begin{align*}
        &\operatorname{trace}\left(\begin{bmatrix}
              \Sigma (x, u) & \0_{n \times m}\\
              \0_{n \times r} & \Lambda (u)
        \end{bmatrix}^\top\right.
        \begin{bmatrix}
                P & - P \\
                - P & P
            \end{bmatrix}
        \left.\begin{bmatrix}
             \Sigma (x, u) & \0_{n \times m}\\
              \0_{n \times r} & \Lambda (u)
        \end{bmatrix}\right)
        \nonumber\\
        %--------------------------------------------
        &\quad = \operatorname{trace} (\Sigma (x, u)^{\top} P \Sigma (x, u))
        + \sigma_u ^2 \operatorname{trace}  \left(\operatorname{diag}(u_t \odot (a - u_t))^\frac{1}{2} \frac{\partial x^\star}{\partial u} (u)^{\top}\right. 
        \left. P \frac{\partial x^\star}{\partial u} (u) \operatorname{diag}(u_t \odot (a - u_t))^\frac{1}{2}\right)
        \nonumber\\
        %--------------------------------------------
        &\quad \leq \sigma_x^2 + \frac{3 \|a\|_2^2}{4} \frac{\ell^2}{c^2} \sigma_u^2 \| P \|_2
        %--------------------------------------------
        = \sigma_x^2 + \frac{3 \|a\|_2^2}{4} \frac{\ell^2}{c^2} \sigma_u^2.
    \end{align*}
Repeating a procedure for deriving~\eqref{pf5:trck_OU_SISC} in Theorem~\ref{thm:trck_OU_SISC}, we have
\begin{align}\label{pf5:trck_JD_new}
            \E \bigl[ \|x_{t} - x^\star (u_t) \|^{2}_{2,P^{1/2}} \bigr]
            &\leq \E \bigl[ \|x_{0}-x^\star (u_0) \|^{2}_{2,P^{1/2}} \bigr]\e^{-2c\alpha t}
             +\frac{1}{\alpha}\left(\frac{\sigma_x^{2}}{2c} + 
            \frac{3 \|a\|_2^2}{4} \frac{\ell^2}{c^2} \frac{\sigma_u^2}{2c}\right)(1-\e^{-2c\alpha t})
            \nonumber\\
            &\quad +\frac{1}{1-\alpha} \frac{1}{2c}\int_{0}^{t}\e^{-2c\alpha(t-\tau)} \E \bigl [\| v(\theta(\tau), u_\tau) \|_{2, P^{1/2}}^2 \bigr]\ d\tau
        \end{align}
for each $\alpha \in (0,1)$.

It remains to compute $\E \bigl [\| v(\theta(t), u_t) \|_{2,P^{1/2}}^{2} \bigr]$. From~\eqref{eq:Hess_JD_SISC} and~\eqref{pf31:trck_JD_new}, we have
\begin{align*}
    \| v(\theta(t), u_t ) \|_{2,P^{1/2}}^{2}
    \le& 2 \left\| \frac{\partial x^\star}{\partial \theta} (u_t) (- c (u_t -\theta(t)) ) \right\|_{2,P^{1/2}}^{2}
    + \frac{1}{2} \sigma_u^4 \left\|\sum_{i=1}^m u_{t,i} (a_i - u_{t,i}) \frac{\partial^2 x^\star (u_t)}{\partial u_{t,i}^2}\right\|_{2,P^{1/2}}^{2}
    \nonumber\\
    %--------------------------------------------
    \le& 2 \ell^2 \| u_t - \theta(t) \|_2^2 + \frac{h_{\mathsf{JD}}^2}{2} \sigma_u^4.
\end{align*}
We compute an upper bound on~$\E[\| u_t - \theta(t) \|_2^2]$. From the It\^o formula, we have
\begin{align*}
    \mathcal{L} \| u_t - \theta(t) \|_2^2
    &= 2 (u_t - \theta(t))^\top (d u_t - \dot \theta(t) dt) 
    + \sigma_u^2 \operatorname{trace} (\operatorname{diag}(u_t \odot (a - u_t))) dt
    \nonumber\\
    %--------------------------------------------
    &= - 2 c \| u_t - \theta(t) \|_2^2 dt 
    -  2 (u_t - \theta(t))^\top \dot \theta(t) dt
    \nonumber\\
    &\quad + \sigma_u^2 \operatorname{trace} (\operatorname{diag}(u_t \odot (a - u_t))) dt
    + 2 \sigma_u (u_t - \theta(t))^\top \operatorname{diag}(u_t \odot (a -u_t))^\frac{1}{2} d\bs_{t}^u
    \nonumber\\
    %--------------------------------------------
    &\le - c \| u_t - \theta(t) \|_2^2 dt + \frac{1}{c} \| \dot \theta(t) \|_2^2 dt + \frac{\| a\|_2^2}{4} \sigma_u^2 dt
    + 2 \sigma_u (u_t - \theta(t))^\top \operatorname{diag}(u_t \odot (a -u_t))^\frac{1}{2} d\bs_{t}^u
\end{align*}
where in the inequality, the Young's inequality and $\operatorname{trace} (\operatorname{diag}(u_t \odot (a - u_t))) \le \frac{\|a\|_2^2}{4}$ are used. Similarly to the proof of Theorem~\ref{thm:NISS}, this leads to
\begin{align*}
    &\E \bigl[ \| u_t - \theta(t) \|_2^2 \bigr]
    \le \e^{-c t} \E \bigl[ \| u_0 - \theta(0) \|_2^2\bigr]
    + \frac{1}{c} \int_0^t \e^{-c (t-\tau)}  \| \dot \theta(\tau) \|_2^2 d\tau 
    + \frac{\| a\|_2^2}{4} \frac{\sigma_u^2}{c} (1 - \e^{-ct}).
\end{align*}
In summary, we have
\begin{align}\label{pf7:trck_JD_new}
    \E \bigl[ \| v(\theta(t), u_t ) \|_{2,P^{1/2}}^{2}\bigr]
     &\le 2\ell^2 \e^{-c t} \E \bigl[ \| u_0 - \theta(0) \|_2^2\bigr] +   \ell^2 \frac{\| a\|_2^2}{2}\frac{\sigma_u^2}{c}  (1 - \e^{-ct}) 
     \nonumber\\
    &\quad +  \frac{2\ell^2}{c} \int_0^t \e^{-c (t-\tau)}  \| \dot \theta(\tau) \|_2^2 d\tau
     + \frac{h_{\mathsf{JD}}^2}{2} \sigma_u^4.
\end{align}
Combining \eqref{pf5:trck_JD_new} and \eqref{pf7:trck_JD_new}, we have~\eqref{eq2:trck_JD_SISC}.
Finally,~\eqref{eq2:trck_JD_SISC} is obtained by taking the limit superior of~\eqref{eq:trck_JD_SISC}.
\end{proof}

%%%%%%%%%%%%%%%%%%%%%%%%%%%%%%%%%%%%%%%%%%%%%%%%%%%%%%%%%%%%%%%%%%
%%%%%%%%%%%%%%%%%%%%%%%%%%%%%%%%%%%%%%%%%%%%%%%%%%%%%%%%%%%%%%%%%%

%%%%%%%%%%%%%%%%%%%%%%%%%%%%%%%%%%%%%%%%%%%%%%%%%%%%%%%%%%%%%%%%%%
%%%%%%%%%%%%%%%%%%%%%%%%%%%%%%%%%%%%%%%%%%%%%%%%%%%%%%%%%%%%%%%%%%

%%%%%%%%%%%%%%%%%%%%%%%%%%%%%%%%%%%%%%%%%%%%%%%%%%%%%%%%%%%%%%%%%%
%%%%%%%%%%%%%%%%%%%%%%%%%%%%%%%%%%%%%%%%%%%%%%%%%%%%%%%%%%%%%%%%%%

\section{Minkowski's Inequality}

In this section $\|\cdot\|$ denotes an arbitrary norm on $\real^n$.

\begin{lemma}[Minkowski's inequality~{\cite[Problems 5.10]{PB:17}}]
  Let $Y$ and $Z$ be random variables with values in $\real^n$ on a
  probability space $(\Omega, {\cal F}, {\mathbb P})$, and let $\E$ denote
  the expectation operator on $\real^n$.  If $\E[\|Y\|^p]^{1/p}$ and
  $\E[\|Z\|^p]^{1/p}$ are finite for some $p \in [1, \infty]$ (where
  $\E[\|\cdot\|^p]^{1/p}$ is read as $\mathrm{ess\,sup}\,\|\cdot\|$ when
  $p=\infty$), then for such $p$,
  \begin{align}\label{eq:Minkowski}
    \E[\|Y + Z\|^p]^{1/p} \le \E[\|Y\|^p]^{1/p} + \E[\|Z\|^p]^{1/p}.
  \end{align}
\end{lemma}

\begin{lemma}[Minkowski-type bound]\label{lem:Minkowski}
Let $X$ and $Y$ be random variables with values in $\real^n$ on a probability
space $(\Omega,\mathcal{F},\mathbb{P})$. Let $b\in\real$ and let
$p\in[1,\infty]$. Assume $\E[\|Y\|^p]<\infty$ and
\begin{equation}\label{eq1:Minkowski-red}
\|X\| \le \|Y\| + |b| \quad \text{a.s.}
\end{equation}
Then $\E[\|X\|^p]<\infty$ and
\begin{equation*}
\E[\|X\|^p]^{1/p} \le \E[\|Y\|^p]^{1/p} + |b|.
\end{equation*}
\end{lemma}
\begin{proof}
The case $p=\infty$ is immediate from~\eqref{eq1:Minkowski-red}, since
taking essential suprema preserves the inequality.
Now let $p\in[1,\infty)$. From~\eqref{eq1:Minkowski-red} and monotonicity
of $t\mapsto t^p$ on $\real_{\geq0}$,
\begin{equation*}
\|X\|^p \le (\|Y\|+|b|)^p \quad \text{a.s.}
\end{equation*}
Taking expectations (which also shows $\E[\|X\|^p]<\infty$) and $p$th roots:
\begin{equation*}
\E[\|X\|^p]^{1/p} \le \E[(\|Y\|+|b|)^p]^{1/p}.
\end{equation*}
Applying Minkowski's inequality to the random variables $\|Y\|$ and the
constant $|b|$ gives
\begin{equation*}
\E[(\|Y\|+|b|)^p]^{1/p}
\le
\E[\|Y\|^p]^{1/p} + |b|.
\end{equation*}
This completes the proof.
\end{proof}

\newpage

%\bibliographystyle{plainurl+isbn}

%\bibliography{alias,Main,FB}
\bibliographystyle{unsrtnat}
\bibliography{alias,Main,FB}

@Book{FB:26-CTDS,
  author =	 {F. Bullo},
  title =	 {Contraction Theory for Dynamical Systems},
  year =	 2026,
  edition =	 {{1.3}},
  publisher =	 {Kindle Direct Publishing},
  ISBN =	 {979-8836646806},
  ISBN-hardocver = {979-8333395283},
  url =		 {https://fbullo.github.io/ctds},
  oldurl =	 {http://motion.me.ucsb.edu/book-ctds},
  funding =	 {FA9550-22-1-0059}
}

@article{AD-VC-AG-GR-FB:23f,
  author =	 {A. Davydov and V. Centorrino and A. Gokhale and G. Russo
                  and F. Bullo},
  title =	 {Time-Varying Convex Optimization: A Contraction and
                  Equilibrium Tracking Approach},
  oldtitle =	 {Contracting Dynamics for Time-Varying Convex
                  Optimization},
  journal =	 tac,
  year =	 2025,
  volume =	 70,
  number =	 11,
  pages =	 {7446-7460},
  arxivdoi =	 {10.48550/arXiv.2305.15595},
  doi =		 {10.1109/TAC.2025.3576043},
  funding =	 {FA9550-22-1-0059}
}

@Article{AG-AD-FB:23o,
  author =	 {A. Gokhale and A. Davydov and F. Bullo},
  title =	 {Contractivity of Distributed Optimization and {Nash}
                  Seeking Dynamics},
  journal =	 lcss,
  year =	 2023,
  volume =	 7,
  pages =	 {3896-3901},
  doi =		 {10.1109/LCSYS.2023.3341987},
  funding =	 {FA9550-22-1-0059},
  olddoi =	 {10.48550/arXiv.2309.05873},
}

@Article{ZM-FB-AGA:23r,
  author =	 {Z. Marvi and F. Bullo and A. G. Alleyne},
  title =	 {Control Barrier Proximal Dynamics: {A} Contraction
                  Theoretic Approach for Safety Verification},
  journal =	 lcss,
  year =	 2024,
  Volume =	 8,
  Pages =	 {880-885},
  doi =		 {10.1109/LCSYS.2024.3402188},
  funding =	 {FA9550-22-1-0059},
  olddoi =	 {10.48550/arXiv.2309.05873}
}

@Article{ZM-FB-AGA:25p,
  author =	 {Z. Marvi and F. Bullo and A. G. Alleyne},
  title =	 {Air Cooled Battery Pack Thermal Management via Control
                  Barrier Proximal Dynamics},
  journal =	 tcst,
  year =	 2025,
  month =	 jun,
  note =	 {Conditionally accepted},
  funding =	 {FA9550-22-1-0059}
}

@Article{ZM-FB-AGA:25ab,
  author =	 {Z. Marvi and F. Bullo and A. G. Alleyne},
  title =	 {Discrete Control Barrier Proximal Dynamics: {Quantized}
                  Multi-Actuator and Sampled-data Systems},
  journal =	 automatica,
  year =	 2025,
  note =	 {Submitted},
  funding =	 {FA9550-22-1-0059}
}

@Article{AC-BG-JC:17,
  title =	 {Saddle-point dynamics: {Conditions} for asymptotic
                  stability of saddle points},
  author =	 {A. Cherukuri and B. Gharesifard and J. Cortes},
  journal =	 sicon,
  year =	 2017,
  volume =	 55,
  number =	 1,
  pages =	 {486-511},
  doi =		 {10.1137/15M1026924},
  fbnote =	 {was AC-BG-JC:15}
}

@ARTICLE{GQ-NL:19,
  author =	 {G. Qu and N. Li},
  title =	 {On the Exponential Stability of Primal-Dual Gradient
                  Dynamics},
  year =	 2019,
  journal =	 csl,
  volume =	 3,
  number =	 1,
  pages =	 {43-48},
  doi =		 {10.1109/LCSYS.2018.2851375},
}

@book{HKK:02,
  AUTHOR =	 "H. K. Khalil",
  TITLE =	 "Nonlinear Systems",
  PUBLISHER =	 ph,
  YEAR =	 2002,
  EDITION =	 3,
  ISBN =	 0130673897,
}

@Article{JC-SKN:19,
  author =	 "J. Cort{\'e}s and S. K. Niederl{\"a}nder",
  title =	 "Distributed Coordination for Nonsmooth Convex
                  Optimization via Saddle-Point Dynamics",
  journal =	 jnls,
  year =	 2019,
  volume =	 29,
  number =	 4,
  pages =	 "1247--1272",
  doi =		 "10.1007/s00332-018-9516-4",
}

@INPROCEEDINGS{JW-NE:11,
  author =	 {J. Wang and N. Elia},
  booktitle =	 cdcecc,
  title =	 {A control perspective for centralized and distributed
                  convex optimization},
  year =	 2011,
  pages =	 {3800-3805},
  address =	 {Orlando, USA},
  doi =		 {10.1109/CDC.2011.6161503},
}

@article{WL-JJES:98,
  Author =	 {W. Lohmiller and J.-J. E. Slotine},
  Journal =	 automatica,
  Number =	 6,
  Pages =	 {683--696},
  Title =	 {On contraction analysis for non-linear systems},
  Volume =	 34,
  Year =	 1998,
  doi =		 {10.1016/S0005-1098(98)00019-3},
}

@ARTICLE{MF-SP-VMP-AR:18,
  author =	 {M. Fazlyab and S. Paternain and V. M. Preciado and
                  A. Ribeiro},
  journal =	 tac,
  title =	 {Prediction-Correction Interior-Point Method for
                  Time-Varying Convex Optimization},
  year =	 2018,
  volume =	 63,
  number =	 7,
  pages =	 {1973-1986},
  DOI =		 {10.1109/TAC.2017.2760256}
}

@Book{CV:09,
  author =	 {C. Villani},
  title =	 {Optimal Transport: Old and New},
  publisher =	 sv,
  year =	 2009,
  doi =		 {10.1007/978-3-540-71050-9}
}

@article{RIK:60,
  title =	 {Monotone operators and convex functionals},
  author =	 {R. I. Kachurovskii},
  journal =	 {Uspekhi Matematicheskikh Nauk},
  volume =	 15,
  number =	 4,
  pages =	 {213--215},
  year =	 1960,
  publisher =	 {Russian Academy of Sciences, Steklov Mathematical
                  Institute of Russian~…}
}

@ARTICLE{AS-EDA-SP-GL-GBG:20,
  author =	 {A. Simonetto and E. {Dall'Anese} and S. Paternain and
                  G. Leus and G. B. Giannakis},
  journal =	 {Proceedings of the IEEE},
  title =	 {Time-Varying Convex Optimization: Time-Structured
                  Algorithms and Applications},
  year =	 2020,
  volume =	 108,
  number =	 11,
  pages =	 {2032-2048},
  doi =		 {10.1109/JPROC.2020.3003156}
}

@article{NMB-JJES:20,
  author =	 {N. M. Boffi and J.-J. E. Slotine},
  title =	 {A Continuous-Time Analysis of Distributed Stochastic
                  Gradient},
  journal =	 {Neural Computation},
  volume =	 32,
  number =	 1,
  pages =	 {36-96},
  year =	 2020,
  doi =		 {10.1162/neco_a_01248},
}

@article{HT-SJC-JJES:21,
  author =	 {H. Tsukamoto and S.-J. Chung and J.-J. E Slotine},
  fullauthor =	 {Hiroyasu Tsukamoto and Soon-Jo Chung and Jean-Jacques
                  E. Slotine},
  doi =		 {10.1016/j.arcontrol.2021.10.001},
  year =	 2021,
  volume =	 52,
  pages =	 {135--169},
  title =	 {Contraction theory for nonlinear stability analysis and
                  learning-based control: {A} tutorial overview},
  journal =	 arc,
}

@article{QCP-NT-JJES:09,
  title =	 {A contraction theory approach to stochastic incremental
                  stability},
  author =	 {Q. C. Pham and N. Tabareau and J.-J. E.  Slotine},
  journal =	 tac,
  volume =	 54,
  number =	 4,
  pages =	 {816--820},
  year =	 2009,
  doi =		 {10.1109/tac.2008.2009619},
}

@ARTICLE{ZA:22,
  author =	 {Z. Aminzare},
  journal =	 csl,
  title =	 {Stochastic Logarithmic {Lipschitz} Constants: {A} Tool to
                  Analyze Contractivity of Stochastic Differential
                  Equations},
  year =	 2022,
  volume =	 6,
  pages =	 {2311-2316},
  doi =		 {10.1109/LCSYS.2022.3148945},
}

@ARTICLE{NKD-SZK-MRJ:19,
  author =	 {N. K. Dhingra and S. Z. Khong and M. R. Jovanovi{\'c}},
  journal =	 tac,
  title =	 {The Proximal Augmented {L}agrangian Method for Nonsmooth
                  Composite Optimization},
  year =	 2019,
  volume =	 64,
  number =	 7,
  pages =	 {2861-2868},
  doi =		 {10.1109/TAC.2018.2867589}
}

@article{SSA-SR:10,
  author =	 {S. S. Ahmad and S. Raha},
  fullauthor =	 {Ahmad, SK. SAFIQUE and Raha, Soumyendu},
  title =	 {On estimation of transient stochastic stability of linear
                  systems},
  volume =	 10,
  DOI =		 {10.1142/s0219493710003017},
  number =	 03,
  journal =	 {Stochastics and Dynamics},
  year =	 2010,
  pages =	 {385-405}
}

@article{FB-IG-AG:12,
  title =	 {Convergence to equilibrium in {Wasserstein} distance for
                  {Fokker}-{Planck} equations},
  volume =	 263,
  DOI =		 {10.1016/j.jfa.2012.07.007},
  number =	 8,
  journal =	 {Journal of Functional Analysis},
  author =	 {F. Bolley and I. Gentil and A. Guillin},
  year =	 2012,
  pages =	 {2430-2457}
}

@book{BO:13,
  title =	 {Stochastic Differential Equations: an Introduction with
                  Applications},
  author =	 {B. {\O}ksendal},
  year =	 2013,
  publisher =	 springer,
  ISBN =	 9783662036204,
  edition =	 6,
}

@article{JB-JJES:19,
  author =	 {J. Bouvrie and J.-J. Slotine},
  fullauthor =	 {Bouvrie, Jake and Slotine, Jean-Jacques},
  title =	 {Wasserstein contraction of stochastic nonlinear systems},
  journal =	 {arXiv preprint arXiv:1902.08567},
  year =	 2019
}

@ARTICLE{EDA-AS-SB-LM:20,
  author =	 {E. Dall'Anese and A. Simonetto and S. Becker and
                  L. Madden},
  journal =	 {IEEE Signal Processing Magazine},
  title =	 {Optimization and Learning With Information Streams:
                  {T}ime-varying algorithms and applications},
  year =	 2020,
  volume =	 37,
  number =	 3,
  pages =	 {71-83},
  doi =		 {10.1109/MSP.2020.2968813}
}

@inbook{FB:08,
  author =	 {Bolley, F.},
  title =	 {Separability and completeness for the {Wasserstein}
                  distance},
  ISBN =	 9783540779131,
  DOI =		 {10.1007/978-3-540-77913-1_17},
  booktitle =	 {Séminaire de Probabilités XLI},
  series =	 {Lecture Notes in Mathematics},
  year =	 2008,
  pages =	 {371-377}
}

@article{CRG-RMS:84,
  title =	 {A class of {Wasserstein} metrics for probability
                  distributions},
  volume =	 31,
  number =	 2,
  journal =	 {Michigan Mathematical Journal},
  author =	 {Givens, C. R. and Shortt, R. M.},
  year =	 1984,
  doi =		 {10.1307/mmj/1029003026}
}

@article{PHAN:21,
  author =	 {P. H. A. Ngoc},
  title =	 {Contraction of stochastic differential equations},
  volume =	 95,
  DOI =		 {10.1016/j.cnsns.2020.105613},
  journal =	 {Communications in Nonlinear Science and Numerical
                  Simulation},
  year =	 2021,
  pages =	 105613
}

@book{IK-SES:14,
  title =	 {Brownian Motion and Stochastic Calculus},
  author =	 {Karatzas, I. and Shreve, S.},
  year =	 2014,
  publisher =	 springer,
  ISBN =	 {978-1-4612-0949-2}
}

@inbook{SC-RK:21,
  title =	 {Modeling the Intraday Electricity Demand in {Germany}},
  ISBN =	 {9783030627324},
  DOI =		 {10.1007/978-3-030-62732-4_1},
  booktitle =	 {Mathematical Modeling, Simulation and Optimization for
                  Power Engineering and Management},
  publisher =	 {Springer International Publishing},
  author =	 {S. Coskun and R. Korn},
  year =	 {2021},
  pages =	 {3-23},
}

@book{HB:11,
  title =	 {Functional Analysis, Sobolev Spaces and Partial
                  Differential Equations},
  ISBN =	 9780387709147,
  DOI =		 {10.1007/978-0-387-70914-7},
  publisher =	 springer,
  author =	 {Brezis, H.},
  year =	 2011
}

@book{PB:17,
  title =	 {Probability and Measure},
  author =	 {P. Billingsley},
  edition =	 3,
  year =	 1995,
  ISBN =	 {0-471-00710-2},
  publisher =	 {John Wiley \& Sons}
}

@article{NT-JJS-QCP:10,
  author =	 {N. Tabareau and J.J. Slotine and Q.-C. Pham},
  fullauthor =	 {Tabareau, Nicolas and Slotine, Jean-Jacques and Pham,
                  Quang-Cuong},
  title =	 {How Synchronization Protects from Noise},
  volume =	 6,
  DOI =		 {10.1371/journal.pcbi.1000637},
  number =	 1,
  journal =	 {PLoS Computational Biology},
  year =	 2010,
  pages =	 {e1000637}
}

@article{ZA-VS:22,
  author =	 {Z. Aminzare and V. Srivastava},
  fullauthor =	 {Aminzare, Zahra and Srivastava, Vaibhav},
  title =	 {Stochastic synchronization in nonlinear network systems
                  driven by intrinsic and coupling noise},
  volume =	 116,
  DOI =		 {10.1007/s00422-022-00928-7},
  number =	 2,
  journal =	 {Biological Cybernetics},
  year =	 2022,
  pages =	 {147-162}
}

@article{AA-JC:25,
  author =	 {A. Allibhoy and J. Cortés},
  fullauthor =	 {Allibhoy, Ahmed and Cortés, Jorge},
  title =	 {Anytime solvers for variational inequalities: The
                  (recursive) safe monotone flows},
  volume =	 177,
  DOI =		 {10.1016/j.automatica.2025.112287},
  journal =	 automatica,
  year =	 2025,
  pages =	 112287
}

@article{PHAN-LTH:26,
  author =	 {P. H. A. Ngoc and L. T. Hieu},
  title =	 {Contraction analysis of stochastic functional
                  differential equations},
  DOI =		 {10.1080/00207721.2026.2625788},
  journal =	 {International Journal of Systems Science},
  year =	 2026,
  pages =	 {1-10}
}

@article{LC-FH-EM-LJR:25,
  title =	 {Monotone Multispecies Flows},
  author =	 {L. Conger and F. Hoffmann and E. Mazumdar and
                  L. J. Ratliff},
  fullauthor =	 {Lauren Conger and Franca Hoffmann and Eric Mazumdar and
                  Lillian J. Ratliff},
  journal =	 {arXiv preprint},
  year =	 2025,
  doi =		 {10.48550/arXiv.2506.22947},
}

@STRING{tcst = "IEEE Transactions on Control Systems Technology"}

@STRING{tac  = "IEEE Transactions on Automatic Control"}

@string{csl = "IEEE Control Systems Letters"}

@string{lcss = "IEEE Control Systems Letters"}

@String{arc = "Annual Reviews in Control"}

@STRING{sicon = "SIAM Journal on Control and Optimization"}

@STRING{automatica = "Automatica"}

@string{science = "Science"}

@STRING{cdcecc = "{IEEE} Conf.\ on Decision and Control and European Control Conference"}

@STRING{sv = "Springer"}

@STRING{springer = "Springer"}

@STRING{PH = "Prentice Hall"}

@STRING{wiley = "John Wiley \& Sons"}

@STRING{jnls    = "Journal of Nonlinear Science"}

\end{document}